\documentclass[11pt,a4paper]{article}
\pdfoutput=1 
\usepackage{fullpage}

\usepackage{algorithm}
\usepackage{algpseudocode}
\usepackage{amsmath}
\usepackage{amsfonts}
\usepackage{amssymb}
\usepackage{amsthm}
\allowdisplaybreaks
\usepackage{mathtools}
\usepackage{thmtools} 

\usepackage[svgnames]{xcolor}
\usepackage{hyperref}
\hypersetup{colorlinks={true},urlcolor={blue},linkcolor={DarkBlue},citecolor=[named]{DarkGreen}}
\usepackage[authoryear,square]{natbib}
\usepackage{enumitem}
\usepackage[capitalise,nameinlink,noabbrev]{cleveref}

\usepackage[T1]{fontenc}
\usepackage[tt=false]{libertine}
\usepackage[libertine]{newtxmath}

\usepackage{doi}
\usepackage{bm}
\usepackage{bbm}
\usepackage{xspace}

\usepackage{booktabs} 
\usepackage{authblk} 
\usepackage{graphicx} 

\usepackage{tcolorbox}

\newtheorem{theorem}{Theorem}[section]
\newtheorem{corollary}[theorem]{Corollary}

\newtheorem{inftheorem}{Informal Theorem}

\theoremstyle{definition}
\newtheorem*{comment*}{Comment}

\newcommand{\bv}{\mathbf{v}}

\newcommand{\bX}{\mathbf{X}}

\newcommand{\calA}{\mathcal{A}}
\newcommand{\calP}{\mathcal{P}}

\usepackage{xcolor}

\usepackage{xspace}
\newcommand{\RR}{\textnormal{\textsc{RoundRobin}}\xspace}
\newcommand{\RRLA}{\textnormal{\textsc{RoundRobin-up-to-the-LastAgent}}\xspace}
\newcommand{\sRRLA}{\textnormal{RRLA}\xspace}
\newcommand{\PRR}{\textnormal{\textsc{Partition-and-RoundRobin}}\xspace}
\newcommand{\sPRR}{\textnormal{PRR}\xspace}
\newcommand{\MFRR}{\textnormal{\textsc{Match\textnormal{\&}Freeze-and-RoundRobin}}\xspace}
\newcommand{\sMFRR}{\textnormal{MFRR}\xspace}

\author[1]{Aris Filos-Ratsikas}
\author[1]{Georgios Kalantzis}
\author[2]{Alexandros A. Voudouris}

\affil[1]{University of Edinburgh, UK}
\affil[2]{University of Essex, UK}

\date{}

\title{\bf Approximate-EFX Allocations with Ordinal and Limited Cardinal Information}

\begin{document}

\maketitle

\begin{abstract}
   We study a discrete fair division problem where $n$ agents have additive valuation functions over a set of $m$ goods. We focus on the well-known $\alpha$-EFX fairness criterion, according to which the envy of an agent for another agent is bounded multiplicatively by $\alpha$, after the removal of \emph{any} good from the envied agent's bundle. The vast majority of the literature has studied $\alpha$-EFX allocations under the assumption that full knowledge of the valuation functions of the agents is available. Motivated by the established literature on the \emph{distortion in social choice}, we instead consider $\alpha$-EFX algorithms that operate under \emph{limited information} on these functions. In particular, we assume that the algorithm has access to the ordinal preference rankings, and is allowed to make a small number of queries to obtain further access to the underlying values of the agents for the goods. We show (near-optimal) tradeoffs between the values of $\alpha$ and the number of queries required to achieve those, with a particular focus on constant EFX approximations. We also consider two interesting special cases, namely instances with a constant number of agents, or with two possible values, and provide improved positive results.    
\end{abstract}

\section{Introduction}

The area of (computational) fair division is concerned with the \emph{fair} allocation of scarce {\em resources} to {\em agents} that have heterogeneous preferences over these resources. Rooted in classic works in economics and mathematics dating back to the late 1940s (e.g., \citep{steinhaus1949division}), and with a strong presence currently in most of the major venues in computer science, the associated literature has studied a plethora of different fair division settings, and has introduced several meaningful notions of fairness. One of the most prominent settings is that of \emph{discrete fair division} (where the resources are indivisible and must irrevocably be allocated to different agents) with the associated notions of fairness being, among others, \emph{relaxations of envy-freeness}. 

Envy-freeness is a quite strong criterion stipulating that the resources (or, {\em goods}) must be allocated in such a way so that no agent would prefer to swap her allocated bundle of goods with that of another agent. While this notion has been quite central in the fair division of divisible goods, very simple examples show that in discrete fair division, envy-free allocations fail to exist. This led researchers in the area, e.g., see \citep{budish2011combinatorial,lipton2004approximately,caragiannis2019unreasonable} to define several relaxations of envy-freeness, with the most popular ones being \emph{envy-freeness up to one good (EF1)} and \emph{envy-freeness up to any good (EFX)}. As the names suggest, these notions stipulate that, if an agent $i$ envies another agent $j$, then this envy of $i$ can be eliminated by the removal of a single good from $j$'s bundle. For EF1, that good can be the favorite one of agent $i$, whereas, for EFX, it has to be her {\em least} favorite.
By definition, EFX is clearly a much fairer, and thus more desirable notion. That being said, while EF1 allocations always exist and can be computed by simple polynomial-time algorithms, the status of EFX allocations is much more enigmatic. Despite intense efforts from an interdisciplinary community, we are still unaware whether EFX allocations exist in general, and, if they do, whether they can be computed in polynomial time; answering this question is arguably the most important open problem in the area of computational fair division. 

While obtaining an answer to this question has turned out to be very challenging, the literature has been rather successful in obtaining positive results for special cases. For example, in a breakthrough work, \citet{chaudhury2024efx} showed that EFX allocations exist for up to three agents, \citet{amanatidis2021maximum} showed that EFX is achievable in polynomial time for bivalued instances in which the agents have up to two different values for the goods, and \citet{christodoulou2023fair} showed similar results for instances that can be encoded by graphs. A significant line of work has considered \emph{$\alpha$-EFX approximations}, where the envy that remains after the removal of any good from the envied agent's bundle is bounded by a multiplicative factor $\alpha \in [0,1]$; e.g., see \citep{amanatidis2024pushing,ANM2019,farhadi2021almost,markakis2023improv,kaviani2025improved,hv2025almost}. Presently, the best known EFX approximation guarantee without any restrictions on the number of agents or their values is $\alpha = \phi-1 = 0.618$ due to \citet{ANM2019}. 

Crucially, all of the aforementioned results are contingent on a rather crucial assumption, namely that the algorithms have \emph{full access to the actual values of the agents for the goods}. More precisely, it is assumed that every agent $i$ has a (\emph{cardinal}) value $v_i(g)$ for each good, and these values are fully and perfectly elicited by the algorithms in their inputs. Although this assumption enables the design of algorithms with good EFX approximations, it neglects the cognitive burden that this elicitation process imposes on the agents, who are asked to express their values for every good on a cardinal scale. It is arguably much more conceivable for agents to provide {\em ordinal} information by ranking the different goods, as well as a limited amount of cardinal information about their actual values that will guide the algorithm's decisions towards the fairest outcomes possible. The premise of this idea is the backbone of the well-established literature of the \emph{distortion} in social choice theory \citep{procaccia2006distortion,distortion-survey}, which measures the deterioration of some desirable social objective due to limited access to the agents' cardinal values. While the objectives typically studied in the context of distortion are aggregate (e.g., the social welfare or the egalitarian welfare), the underlying principle can be applied verbatim to individual objectives such as the fairness objectives discussed above. Motivated by this, the main question that we aim to answer is the following:

\begin{quote}
\emph{What is the best possible EFX approximation achievable by algorithms that only use ordinal information, and possibly also a small additional amount of cardinal information?}
\end{quote}

\noindent Below, we first present the model that we will study in our work, and then state our main results related to the question above. 

\subsection{The Model}
We consider a discrete fair division setting with a set $N=\{1,\ldots,n\}$ of $n \geq 2$ {\em agents} and a set $G$ of $m \geq 2$ {\em goods}. 
Each agent $i \in N$ has an additive {\em valuation function} $v_i:G \rightarrow \mathbb{R}_{\geq 0}$ which assigns a non-negative value $v_i(g)$ to each good $g \in G$. 
Let $\bv = (v_i)_{i \in N}$ be the {\em valuation profile} consisting of the valuation functions of all agents.  
The value of agent $i$ for a {\em bundle} $X \subseteq G$ of goods is $v_i(X) = \sum_{g \in X} v_i(g)$. A {\em partial allocation} $\bX = (X_i)_{i \in N}$ is a partition of a subset of the goods into $n$ bundles such that $X_i \cap X_j = \varnothing$. 
A complete allocation, or simply, {\em allocation}, is a partial allocation such that $\bigcup_{i \in N} X_i = G$. 

We are interested in computing allocations that are (approximately) fair according to (relaxations of) {\em envy-freeness}. In particular, an allocation $\bX$ is {\em envy-free} (EF) if, for any two agents $i$ and $j$, $v_i(X_i) \geq v_i(X_j)$. As already mentioned, this fairness criterion cannot be satisfied even for simple instances, and we thus focus on two of its most-studied relaxations.
Let $\alpha$ be any real number in $[0,1]$, and let $\bX$ be an allocation of the goods to the agents. Then:
\begin{itemize}
    \item $\bX$ is $\alpha$-{\em envy-free up to one good} ($\alpha$-EF1), if for any agent $i$ and agent $j\neq i$ such that $X_j \neq \varnothing$, it holds that $v_i(X_i) \geq \alpha \cdot v_i(X_j\setminus \{g\})$ for \emph{some} $g \in X_j$. If $\alpha=1$, then the allocation is EF1.

    \item $\bX$ is $\alpha$-{\em envy-free up to any good} ($\alpha$-EFX), if for any agent $i$ and agent $j\neq i$ such that $X_j \neq \varnothing$, it holds that $v_i(X_i) \geq \alpha \cdot v_i(X_j\setminus \{g\})$ for \emph{any} $g \in X_j$. If $\alpha=1$, then the allocation is EFX.
\end{itemize}

An algorithm $\calA$ takes as input some information $I(\bv)$ about the valuation profile $\bv$ and outputs an allocation $\bX=\calA(I(\bv))$. 
In contrast to previous work on this area, where complete access to the valuation functions is assumed, i.e., $I(\bv)=\bv$, here we focus on algorithms that are given some \emph{limited} information about the valuations. In particular, we assume that $\mathcal{A}$ has access to the (ordinal) preference rankings $(\succ_{i})_{i \in N}$ of the agents that are \emph{consistent} with $\bv$, i.e., $v_i(g) \geq v_i(g') \Rightarrow g \succ_i g'$ (with ties being broken arbitrarily according to some fixed tie-breaking rule). Additionally, $\mathcal{A}$ can obtain some further restricted access to $\bv$ via a (limited) number of queries. Specifically, we adopt the query model of \citet{amanatidis2021peeking}, which was originally used in the context of distortion in utilitarian voting,  where the algorithm can perform a number of \emph{value queries} to the agents; such a query inputs an agent $i$ and a good $g$, and outputs the value $v_i(g)$ of the agent for that good. The objective is to achieve the best possible tradeoff between $\alpha$ (the EF1/EFX approximation factor) and the number of value queries per agent. We will refer to algorithms that only use the preference rankings of the agents as \emph{ordinal algorithms}, and to algorithms that are allowed to make additional queries as \emph{query-enhanced algorithms}. 

\subsection{Our Contribution}
As a warm-up, in \cref{sec:ordinal}, we consider the first part of our main question regarding purely ordinal algorithms. Unsurprisingly, it turns out that such algorithms are bound to perform poorly in terms of EFX approximations. We prove a tight bound, captured by the following informal theorem.\footnote{All of our results assume that $m \geq n$; this is without loss of generality, as otherwise a trivial (ordinal) algorithm that assigns at most one arbitrary good to each agent achieves (exact) EFX.}

\begin{inftheorem}\label{infthm:ordinal-algorithms}
The best possible EFX approximation achievable by any ordinal algorithm is $1/(m-n)$.  
\end{inftheorem}

\noindent Driven by this (expected) poor performance of ordinal algorithms, we then turn our attention to the second part of our main question and consider \emph{query-enhanced algorithms} in \cref{sec:binary-search}. For these, more powerful algorithms, we focus on quantifying the number of queries required to achieve constant-EFX approximations but do not aim to fully optimize these constants (even though, we do provide specific bounds in some cases). In other words, our goal is to answer the following question. 


\begin{quote}
\emph{
How many queries per agent are sufficient and necessary to achieve $\alpha$-EFX for constant $\alpha$?}
\end{quote}
For unrestricted instances with additive valuations, we manage to provide a clean qualitative answer to the question above, captured by the following informal theorem.

\begin{inftheorem}\label{infthm:polylogarithmic}
A number of queries per agent that is polylogarithmic in the number of goods is both necessary and sufficient to achieve $\alpha$-EFX for constant $\alpha$.
\end{inftheorem}

Unfolding the details of \cref{infthm:polylogarithmic}, our asymptotic upper bound on the number of queries is $O(n+\log^2 m)$ and is achieved by a novel algorithm which we present in \cref{sec:binary-search}; see \cref{cor:general-positive-constant}. Our algorithm works by first learning the values of each agent for their top $n-1$ goods via queries. Subsequently, using an appropriately defined set of threshold parameters, it performs $O(\log{m})$ binary searches (each requiring $O(\log{m})$ queries) to partition the goods into buckets of different value, and thus create a \emph{virtual valuation function} for each agent, which serves as a proxy for her real, unknown valuation function. Finally, the algorithm calls any (fully-informed) constant-EFX algorithm as a black-box, with the virtual valuation functions as input, and outputs the returned allocation. To obtain the desired constant EFX bound, we essentially bound the loss in the EFX approximation due to using the virtual valuation functions. The lower bound on the number of queries of \cref{infthm:polylogarithmic} (\cref{cor:general-negative-contant}) actually establishes that, for any constant $\alpha$, strictly more than (asymptotically) $\frac{\log m}{\log\log m}$ queries (i.e., $o\left(\frac{\log{m}}{\log\log{m}}\right)$ queries) are required to achieve $\alpha$-EFX. Observe that, for any $\varepsilon >0$, this number is asymptotically larger than $(\log{m})^{1-\varepsilon}$, and hence, when $n = O(\log^2 m)$, the bound is tight up to an almost logarithmic factor. We remark that both the upper and the lower bounds of \cref{infthm:polylogarithmic} are in fact corollaries of more general, parameterized bounds that we prove in \cref{sec:binary-search}; see \cref{thm:virtual}. \medskip

\noindent We also obtain improved $\alpha$-EFX approximations for two special cases, namely \emph{instances with a constant number of agents $n$} and \emph{bivalued instances}.

\medskip
\noindent 
{\bf Improved bounds for constant number of agents.} Instances with constant $n$ are very natural in practice, and, in fact, the state-of-the-art results in terms of exact or approximate-EFX allocations are established for such cases \citep{chaudhury2024efx,amanatidis2024pushing}. 
For constant $n$, our general parameterized bounds above do not preclude the existence of algorithms that might be able to compute  $\Omega\left(\sqrt{k}\cdot m^{-1/(2k-1)}\right)$-EFX allocations using $k$ queries per agent, whereas the algorithm we design in \cref{sec:binary-search} only achieves $\Omega\left(m^{-\log m / (k + \log m)}\right)$-EFX with $k$ queries. In \cref{sec:constant-n}, we propose a different algorithm, which achieves a tighter bound of $\Omega\left(\frac{1}{\sqrt{k}}\cdot m^{-1/(2k-1)}\right)$-EFX (\cref{cor:constant-n-general-k}). When both $k$ and $n$ are constants, this bound is tight (\cref{cor:constant-n-constant-k}). For general (non-constant) $k$, the EFX approximation guarantee of this algorithm approaches the best possible upper bound established by our impossibility results. In particular, with $O(\log m)$ queries per agent, it achieves near-constant EFX approximation, in particular, $\Omega(1/\sqrt{\log{m}})$; see \cref{cor:constant-n-k=logm}. 

\medskip
\noindent 
{\bf Improved bounds for bivalued instances.} 
Bivalued instances capture well-structured settings in which all agents have two (possibly different) values for the goods. Such instances were first studied in the fair division literature and for the EFX notion in particular by \citet{amanatidis2021maximum}. Since then, the study of bivalued instances and their variants has gained significant attention, leading to a plethora of works; e.g., see \citep{jin2025pareto,byrka2025probing,lin2025approximately,garg2023computing,babaioff2021fair}. This particular structure of the unknown valuation functions allows us to show a crisper bound of $O(\log{n})$ on the number of queries that suffice to achieve a constant EFX approximation (\cref{thm:bivalued-logn}). This is achieved by an algorithm which carefully combines algorithms that have been proposed in previous work, in particular, \RR and \textsc{Match}\&\textsc{Freeze}~\citep{amanatidis2021maximum,jin2025pareto}. Observe that this bound relies only on $n$, which is much smaller than $m$ in most applications.

\subsection{Related Work and Discussion}
The study of EFX allocations was initiated by \citet{caragiannis2019unreasonable}, as a fairer alternative to EF1 allocations, which were explicitly defined by \citet{budish2011combinatorial} but implicitly also studied before by \citet{lipton2004approximately}; see they survey of \citet{amanatidis2023fair} for an overview. As mentioned earlier, $\alpha$-EFX allocations have been in the forefront of research in computational fair division in recent years, with the best possible approximation currently known being $0.618$ due to \citet{ANM2019}, see also \citep{farhadi2021almost}. Improved approximations for important special cases have also been considered, such as a small number of (types of) agents \citep{amanatidis2024pushing,hv2025almost,chaudhury2024efx,akrami2025efx,hv2025efx}, instances on (multi)graphs \citep{amanatidis2024pushing,christodoulou2023fair,zeng2025structure,kaviani2025improved,sgouritsa2025existence}, and instances with a small number of possible values \citep{amanatidis2021maximum,amanatidis2024pushing,jin2025pareto,byrka2025probing,garg2023computing}. In the latter category, of particular importance to our work are the results on bivalued instances, for which (exact) EFX allocations are achievable. This was established first by \citet{amanatidis2021maximum} when all agents have the same two values, and was recently extended to the case of personalized bivalued instances independently by \citet{jin2025pareto} and \citet{byrka2025probing}. In addition, \citet{garg2023computing} showed that, for such instances, EFX and Pareto optimality can be achieved simultaneously. 

The model of query-enhanced ordinal algorithms that we consider was introduced by \citet{amanatidis2021peeking} for the study of the {\em distortion} in social choice theory; the distortion typically quantifies the deterioration of an aggregate objective, such as the social welfare, due to limited information about the preferences of the agents. Since its introduction, this model has been adopted by a series of works, see \citep{amanatidis2022matching,amanatidis2024dice,caragiannis2024beyond,ebadian2025bit}; that said, to the best of our knowledge, value query-enhanced ordinal algorithms had not been studied in the context of fair division prior to our work. It is worth to note that a small number of fairly recent works consider discrete fair division settings under query access to the agents' values, but, crucially, the models and investigations in those papers are markedly different from ours. Conceptually closest to ours is the work of \citet{oh2021fairly}, who considered fair division algorithms that access the preferences of the agents via value queries, but, contrary to our setting, in their investigation the algorithms are not aware even of the ordinal rankings of the agents. As a result, they are mainly concerned with EF1 allocations, and consider EFX allocations only for two agents. Naturally, since their algorithms are more restricted, their bounds are much weaker than ours. Quite recently, \citet{bu2024logarithmic} studied a setting where the algorithms make queries to compare the values of the agents for different bundles, and show $\Theta(\log m)$ query bounds for the notions of EF1, PROP1, and $1/2$-MMS. 

In other, less related works, \citet{Feige2025low} considered the communication complexity of computing allocations that are MMS, EF1, or a relaxation of PROP1, measuring the number of \emph{bits} of information that the algorithm needs to elicit, rather than the number of queries. \citet{halpern2021fair} studied algorithms that take as input top-$k$ preference rankings, and aim to find the value of $k$ for which EF1 or MMS allocations can be achieved. The authors also consider the loss in social welfare (known as {\em price of fairness} in this context) due to the lack of expressiveness of the preferences. \citet{benade2022dynamic} considered a setting of dynamic (online) fair division, in which the goods arrive sequentially and only ordinal information about the preferences of the agents is revealed to the algorithm, and provide fairness guarantees with high probability when the values are drawn from known distributions. Finally, we remark that some works, e.g., \citep{bouveret2010fair,li2022proportional} study fair division settings with ordinal preferences, but, unlike ours, their goal is not to approximate a cardinal objective.

\section{Warm-up: Ordinal Algorithms}\label{sec:ordinal}
We start the presentation of our technical results by considering {\em ordinal} algorithms which have access to the {\em ordinal preference profile} $\succ=(\succ_i)_{i \in N}$ of the agents. To be exact, for each agent $i \in N$, such an algorithm takes as input a ranking $\succ_i$ of the goods that is consistent with the valuation function of the agent, that is, for any two goods $g_1, g_2 \in G$,  $g_1 \succ_i g_2$ implies that $v_i(g_1) \geq v_i(g_2)$.  \medskip

Before we present our main results for EFX, we remark that exact EF1 allocations can be computed using only such ordinal information. 
This is achieved by the \RR algorithm (\cref{alg:round-robin}), which is ordinal, and known to output EF1 allocations \citep{lipton2004approximately}. The \RR algorithm will be quite useful in the following sections, where it will be used as a subroutine.

\begin{algorithm}[!ht]
\caption{\RR}\label{alg:round-robin}
\begin{algorithmic}[1]
\small
\For {each agent $i \in N$}
    \State $X_i \gets \varnothing$ \Comment{Empty initial allocation.}
\EndFor
\State $\calP \gets G$ \Comment{Pool of available goods, initially all are available.}
\While{$\calP \neq \varnothing$} \Comment{As long as there are available goods, repeat.}
\For{each agent $i \in [n]$} \Comment{Consider the agents one by one in an arbitrary but fixed order.}
    \If{$\calP = \varnothing$} \Comment{If the pool of available goods is empty,}
        \State \textbf{break}   \Comment{stop.}
    \EndIf
    \State $g^*_i \gets $ top-ranked good in $\calP$ according to $\succ_i$ \Comment{Find $i$ top-ranked available good,}
    \State $X_i \gets X_i \cup \{g^*_i\}$. \Comment{add it to $i$'s bundle,}
    \State $\calP \gets \calP \setminus \{g^*_i\}$. \Comment{and remove it from the pool.}
\EndFor
\EndWhile
\State \Return $\bX = (X_1,\ldots,X_n)$.
\end{algorithmic}
\end{algorithm}

We now consider $\alpha$-EFX allocations and show a tight bound on the best possible value of $\alpha$ that can be achieved by ordinal algorithms. Unsurprisingly, this value is quite small due to the very limited amount of information that is used. We will explore how this bound can be improved by utilizing some additional, small amount of cardinal information in subsequent sections. We first present the positive result, which is achieved by a rather simple algorithm, coined \RRLA (\sRRLA); see \cref{alg:round-robin-and-rest}. As its name suggests, the algorithm performs one round of \RR for the first $n-1$ agents, and then assigns all the remaining unallocated goods to the last agent. The algorithm is clearly ordinal as it only requires the ordering of the goods to operate.

\begin{algorithm}[!ht]
\caption{\RRLA (\sRRLA)}\label{alg:round-robin-and-rest}
\begin{algorithmic}[1]
\small
\For {each agent $i \in N$}
    \State $X_i \gets \varnothing$ \Comment{Empty initial allocation.}
\EndFor
\State $\calP \gets G$ \Comment{Pool of available goods, initially all are available.}
\For{each agent $i \in \{1,\ldots,n-1\}$} \Comment{For all agents but the last one}
    \If{$\calP = \varnothing$} \Comment{If the pool of available goods is empty,}
        \State \textbf{break}   \Comment{stop.}
    \EndIf
    \State $g^*_i \gets $ top-ranked good in $\calP$ according to $\succ_i$ \Comment{Find $i$ top-ranked available good,}
    \State $X_i \gets X_i \cup \{g^*_i\}$. \Comment{add it to $i$'s bundle,}
    \State $\calP \gets \calP \setminus \{g^*_i\}$. \Comment{and remove it from the pool.}
\EndFor
\State $X_n \gets \calP$ \Comment{Agent $n$ receives all the remaining goods from the pool.}
\State \Return $\bX = (X_1,\ldots,X_n)$.
\end{algorithmic}
\end{algorithm}


\begin{theorem}
The \RRLA algorithm (\cref{alg:round-robin-and-rest}) computes a $\frac{1}{m-n}$-EFX allocation. 
\end{theorem}

\begin{proof}
To bound the EFX approximation achieved by the algorithm, first observe that $|X_i|\leq 1$ for any agent $i \in [n-1]$, and thus every agent $j \in N$ is EFX towards $i$. Now consider agent $n$, who receives all the remaining goods, and consider any other agent $i \in [n-1]$. Agent $i$ was allocated her (single) good when all of the goods in $X_{n}$ were in the pool $\calP$, and hence $v_i(X_i) \geq v_i(g)$ for any $g \in X_n$. Since $|X_n| \leq m-n+1$, it follows that $v_i(X_i) \geq \frac{1}{m-n} \cdot v_i(X_n\setminus\{g\})$ for any $g \in X_n$, and we obtain the desired EFX approximation.
\end{proof}

We next show the matching upper bound. 

\begin{theorem}
For any $m > n+2$, no ordinal algorithm can compute an $\alpha$-EFX allocation for any $\alpha<1/(m-n)$. 
\end{theorem}

\begin{proof}
We will consider different valuation profiles that induce the same preference profile $\succ$ in which the rankings of all agents are identical, namely $\succ_i = \succ_j$ for all $i,j \in N$, with $\succ_i = g_1 \succ g_2 \succ \ldots \succ g_m$. Since the algorithm is ordinal, it cannot differentiate between the different valuation profiles, which, depending on the allocation $\bX$ computed by the algorithm, will allow us to choose an appropriate valuation profile $\bv$ that is consistent with $\succ$ and show an EFX approximation of at most $1/(m-n)$. 

Consider the top-$(n-1)$ ranked goods $\{g_1,\ldots,g_{n-1}\}$. Since there are $n$ agents, at least one of them is not assigned any of these goods; without loss of generality, suppose that agent $n$ is such an agent. We switch between the following two cases.

\medskip
\noindent{\bf Case 1: There is a good $g^* \in \{g_1,\ldots,g_{n-1}\}$ and an agent $i \in [n-1]$ such that $g \in X_i$ and $|X_i|\geq 2$.}
Consider a valuation profile $\bv$ consistent with $\succ$ such that, for every agent $j \in N$,
\begin{align*}
    v_j(g) = 
    \begin{cases}
        1, \text{ if } g \in \{g_1,\ldots,g_{n-1}\} \\
        0, \text{ otherwise.}
    \end{cases}
\end{align*}
Then, for agent $n$, since $X_n \cap \{g_1,\ldots,g_{n-1}\} = \varnothing$, we have that $v_n(X_n) = 0$ and there is a good $g \in X_i \setminus\{g^*\}$ such that $v_n(X_i \setminus \{g\}) = 1$. This implies that $\bX$ is $0$-EFX. 

\medskip
\noindent{\bf Case 2: For every good $g \in \{g_1,\ldots,g_{n-1}\}$, there an $i\in[n-1]$ such that $X_i = \{g\}$.}
Consequently, agent $n$ is allocated all $m-n+1$ goods besides those in $\{g_1,\ldots,g_{n-1}\}$. 
Consider then a valuation profile $\bv$ consistent with $\succ$ such that $v_j(g)=1$ for any agent $j \in N$ and good $g \in G$. For any agent $i \in [n-1]$, we have $v_i(X_i)=1$ and $v_i(X_n \setminus \{g\}) = n-m$ for any good $g \in X_n$. This implies that $\bX$ is $1/(m-n)$-EFX and the proof is complete. 
\end{proof}

\section{Query-Enhanced Algorithms} \label{sec:binary-search}
We now turn our attention to a class of {\em query-enhanced algorithms} that have access to the preference profile $\succ$ of the agents and are allowed an additional number $k \in [m]$ of \emph{value queries} per agent, which they can use to acquire more accurate information about the valuation functions. Formally, a value query takes as input an agent $i$ and a good $g$, and returns the value $v_i(g)$ of $i$ for $g$. Clearly, an algorithm that can make $k=m$ queries per agent is equivalent to an algorithm that has complete access to the valuation functions of the agents, whereas an algorithm that makes $k=0$ queries is an ordinal algorithm. Our goal is to design algorithms that make a small number of queries and compute $\alpha$-EFX allocations with an $\alpha$ as large as possible, regardless of the unknown part of the valuation functions of the agents.

Our first result in this setting is that by making $O(n+k\log{m})$ queries per agent, it is possible to achieve a $\frac{\rho}{2m^{1/(k+1)}}$-EFX allocation, where $\rho$ is the EFX approximation achieved by a fully-informed algorithm (i.e., one which has access to the values of all agents for all goods), which we use as a \emph{black-box}. 
In more detail, our algorithm, coined \textsc{$\rho$-Virtual-EFX} (see \cref{alg:algorithm:virtual}) works as follows: It first defines $k+1$ thresholds, $\theta_\ell = m^{-\ell/(k+1)}$ for $\ell \in \{0,\ldots,k\}$, and 
then it computes a {\em virtual} valuation function for every agent $i$ based on these thresholds. 
In particular, we first learn agent $i$'s top $n-1$ values $v_{i,1}, \ldots, v_{i,n-1}$ by querying for the goods that agent $i$ ranks in the first $n-1$ positions. 
Afterwards, for each $\ell \in [k]$, the algorithm performs binary search to find the set $S_\ell$ of all goods that agent $i$ values at least $v_{i,n-1}\cdot\theta_\ell$. The virtual value of all these goods is then defined to be this lower bound, $v_{i, n-1} \cdot \theta_\ell$. The algorithm outputs the allocation that is computed by the full-information $\rho$-EFX algorithm when given as input the virtual valuation functions. Our main result implies several interesting corollaries by choosing the value of $k$ appropriately.


\begin{algorithm}
\caption{\textsc{$\rho$-Virtual-EFX}}\label{alg:algorithm:virtual}
\begin{algorithmic}[1]
\small
\For{each $\ell \in \{0,\ldots,k\}$}
    \State $\theta_\ell \gets m^{-\ell/(k+1)}$
\EndFor
\For{every agent $i$}
    \State Query agent $i$ for her top-$(n-1)$ goods $g_{i,1}, \ldots, g_{i,n-1}$ to learn the values $v_{i,1}, \ldots, v_{i,n-1}$.
    \For{each $\ell \in [k]$}
        \State Run a binary search to learn a set of goods 
        $S_{i,\ell} \gets  \{g: v_{i,n-1}\cdot \theta_\ell \leq v_i(g) < v_{i,n-1} \theta_{\ell-1} \}$.
    \EndFor
    \State $S_{i,k+1} \gets$ set of remaining goods
    \State Define the virtual valuation $\tilde{v}_i$ such that 
        \begin{align*}
            \tilde{v}_i(g) \gets 
            \begin{cases}
                v_{i,\ell} & \text{if $g=g_{i,\ell}, \ell\in [n-1]$} \\
                v_{i,n-1} \theta_\ell & \text{if $g \in S_{i,\ell}, \ell \in [k]$} \\
                0 & \text{if $g \in S_{i,k+1}$}.
            \end{cases}
        \end{align*}
\EndFor
    \State Run the $\rho$-EFX algorithm with the virtual valuations as input, and return this allocation $\bX = (X_i)_i$. 
\end{algorithmic}
\end{algorithm}

\begin{theorem} \label{thm:virtual}
For any $k \geq 1$, given any $\rho$-EFX full-information algorithm, the \textnormal{\textsc{$\rho$-EFX-Virtual}} algorithm (\cref{alg:algorithm:virtual}) makes $O(n+k\log{m})$ queries per agent and computes a $\frac{\rho}{2m^{1/(k+1)}}$-EFX allocation.
\end{theorem}

\begin{proof}
For any agent $i$, the algorithm queries the top-$(n-1)$ values of $i$ and then runs $k$ binary searches for the remaining $m-n$ goods to construct the sets $S_\ell$ for $\ell \in [k+1]$. So, the algorithm indeed makes $O(n+k\log{m})$ queries. 

We now focus on showing the EFX approximation bound. First observe that, by the definition of the $\theta$ thresholds, for any $\ell \in [k]$, we have 
$$\frac{\theta_{\ell-1}}{\theta_\ell} = \frac{ m^{-(\ell-1)/(k+1)}}{ m^{-\ell/(k+1)}} = m^{1/(k+1)} = \theta_1^{-1}.$$ 
Without loss of any generality, we will take the perspective of agent $1$ and show that she is approximately-EFX towards any other agent $i\in\{2,\ldots,n\}$; clearly, if $X_i$ is a singleton, agent $1$ is EFX towards $i$, and we can thus focus on the case where $X_i$ includes at least two goods. 

For any agent $i \in [n]$ and integer $\ell \in [k+1]$, let $\alpha_{i,\ell} = |X_i\cap S_{1,\ell}|$ be the number of goods that agent $i$ gets from the set $S_{1,\ell}$. Since agent $1$ is $\rho$-EFX towards any agent $i \in \{2,\ldots,n\}$ for the virtual valuation $\tilde{v}_1$, by the definition of the algorithm, we have that $\tilde{v}_1(X_1) \geq \rho \cdot \tilde{v}_1(X_i\setminus\{g\})$ for any $g \in X_i$. Let $I_E$ be the indicator variable that is $1$ if the event $E$ is true, and $0$ otherwise. 
For the virtual valuation of agent $1$, for any agent $i \in [n]$ (including $1$), we have
\begin{align*}
    \tilde{v}_1(X_i) = \sum_{\ell=1}^{n-1} I_{g_{1,\ell} \in X_i} \cdot v_{1,\ell} + \sum_{\ell=1}^k \alpha_{i,\ell}\cdot v_{1,n-1}\theta_\ell 
\end{align*}
For the true valuation of agent $1$, since the virtual values are at most the true values, we have
\begin{align*}
    v_1(X_1) \geq \tilde{v}_1(X_1)
\end{align*}
and, for any $i \in \{2,\ldots,n\}$, 
\begin{align*}
    v_1(X_i) \leq \sum_{\ell=1}^{n-1} I_{g_{1,\ell} \in X_i} \cdot v_{1,\ell} + \sum_{\ell=1}^k \alpha_{i,\ell}\cdot v_{1,n-1} \theta_{\ell-1} 
    + \alpha_{i,k+1}\cdot v_{1,n-1} \theta_k. 
\end{align*}
Since $\theta_{\ell-1} = \frac{\theta_{\ell-1}}{\theta_\ell} \cdot \theta_\ell = \theta_1^{-1}\cdot \theta_\ell$ for any $\ell \in [k]$, and $1 \leq \theta_1^{-1}$, we can rewrite the upper bound on $v_1(X_i)$ as 
\begin{align*}
    v_1(X_i) 
    &\leq \theta_1^{-1} \bigg( \sum_{\ell=1}^{n-1} I_{g_{1,\ell} \in X_i} \cdot v_{1,\ell} + \sum_{\ell=1}^k \alpha_{i,\ell}\cdot v_{1,n-1} \theta_\ell \bigg) + \alpha_{i,k+1}\cdot v_{1,n-1} \theta_k \\
    &= \theta_1^{-1} \cdot \tilde{v}_1(X_i) + \alpha_{i,k+1}\cdot v_{1,n-1} \theta_k.
\end{align*}
Using the $\rho$-EFX property of the virtual valuation function, that is, $\tilde{v}_1(X_i) \leq \frac{1}{\rho} \tilde{v}_1(X_1) + \min_{g \in X_i} \tilde{v}_1(g)$, we finally obtain
\begin{align*}
    v_1(X_i) 
    &\leq  \theta_1^{-1} \cdot \bigg( \frac{1}{\rho} \tilde{v}_1(X_1) + \min_{g \in X_i} \tilde{v}_1(g) \bigg) + \alpha_{i,k+1}\cdot v_{1,n-1} \theta_k.
\end{align*}
The EFX approximation is
\begin{align} \label{eq:virtual:efx-approximation}
    \frac{v_1(X_1)}{v_1(X_i)} \geq \frac{ \tilde{v}_1(X_1) }{ \theta_1^{-1} \cdot \bigg( \frac{1}{\rho} \tilde{v}_1(X_1) + \min_{g \in X_i} \tilde{v}_1(g) \bigg) + \alpha_{i,k+1}\cdot v_{1,n-1} \theta_k. }
\end{align}
Observe that the last expression is increasing in $\tilde{v}_1(X_1)$. 

We will now argue by induction that, in the worst case, the allocation $\bX$ computed by the algorithm is such that each of the top-$(n-2)$ goods $\{g_{1,1}, \ldots, g_{1,n-2}\}$ of agent $1$ is assigned to a different agent from the set $J = [n]\setminus\{1,i\}$, and each of these agents is only assigned that good, i.e., $|X_j|=1$ for any $j \in J$. In particular, there is an ordering $\{j_1, \ldots, j_{n-2}\}$ of the agents in $J$ such that $X_{j_\ell} = \{g_{1,\ell} \}$ for any $\ell \in [n-2]$.

\medskip
\noindent 
{\bf Base Case:} good $g_{1,1}$. 
\begin{itemize}
    \item If $g_{1,1} \in X_1$, then, by definition, $\tilde{v}_1(X_1) \geq v_{1,1}$. 
    \item If $g_{1,1} \in X_i$, then, since $|X_i|\geq 2$ and agent $1$ is $\rho$-EFX towards agent $i$ for $\tilde{v}_1$, $\tilde{v}_1(X_1) \geq \rho \cdot v_{1,1}$.
    \item If $g_{1,1} \in X_{j_1}$ for some $j_1 \in [n]\setminus\{1,i\}$ and $|X_{j_1}|\geq 2$, then, since agent $1$ is $\rho$-EFX towards agent $j_1$ for $\tilde{v}_1$, $\tilde{v}_1(X_1) \geq \rho \cdot v_{1,1}$.
\end{itemize}
In all these cases, we have $\tilde{v}_1(X_1) \geq \rho \cdot v_{1,1}$, and \eqref{eq:virtual:efx-approximation} gives us that the EFX approximation can be bounded as follows:
\begin{align*}
    \frac{v_1(X_1)}{v_1(X_i)} \geq \frac{ \rho \cdot v_{1,1} }{ \theta_1^{-1} \cdot \bigg( v_{1,1} + \min_{g \in X_i} \tilde{v}_1(g) \bigg) + \alpha_{i,k+1}\cdot v_{1,n-1} \theta_k }.
\end{align*}
If $\min_{g \in X_i} \tilde{v}_1(g) > 0$, then, since $\tilde{v}_1(g)=0$ for any  $g\in S_{1,k+1}$, it must be the case that $\alpha_{i,k+1} = 0$. Using also the facts that $\theta_1^{-1} = m^{1/(k+1)}$ and $\min_{g \in W_i} \tilde{v}_1(g) \leq v_{1,1}$, we have 
\begin{align*}
    \frac{v_1(X_1)}{v_1(X_i)} \geq \frac{ \rho \cdot v_{1,1} }{ \theta_1^{-1} \cdot \big( v_{1,1} + v_{1,1} \big)} = \frac{\rho}{2m^{1/(k+1)}}.
\end{align*}
Otherwise, since $\alpha_{i,k+1}\leq m$, $v_{1,n-1} \leq v_{1,1}$ and $\theta_k = m^{-k/(k+1)}$, we have
\begin{align*}
    \frac{v_1(X_1)}{v_1(X_i)} 
    &\geq \frac{ \rho \cdot v_{1,1} }{ \theta_1^{-1} \cdot v_{1,1} + \alpha_{i,k+1}\cdot v_{1,n-1} \theta_k } \\
    &= \frac{ \rho \cdot v_{1,1} }{ m^{1/(k+1)} \cdot v_{1,1}  + m\cdot v_{1,1} m^{-k/(k+1)} }
    = \frac{ \rho }{ 2 m^{1/(k+1)} }.
\end{align*}
So, it must be the case that $X_{j_1} = \{g_{1,1}\}$.

\bigskip
\noindent 
{\bf Inductive Hypothesis:} For some $\lambda \in \{1,\ldots,n-3\}$, $X_{j_\ell} = \{g_{1,\ell}\}$ for any $\ell \in [\lambda]$.

\bigskip
\noindent 
{\bf Induction Step:} good $g_{1,\lambda+1}$.
\begin{itemize}
    \item If $g_{1,\lambda+1} \in X_1$, then, by definition, $\tilde{v}_1(X_1) \geq v_{1,\lambda+1}$. 
    \item If $g_{1,\lambda+1} \in X_i$, then, since $|X_i|\geq 2$ and agent $1$ is $\rho$-EFX towards agent $i$ for $\tilde{v}_1$, $\tilde{v}_1(X_1) \geq \rho \cdot v_{1,\lambda+1}$.
    \item If $g_{1,\lambda+1} \in W_{j_1}$ for $j_{\lambda+1} \in [n]\setminus\{1,i,j_1, \ldots,j_\lambda\}$ and $|X_{j_{\lambda+1}}|\geq 2$, then, since agent $1$ is $\rho$-EFX towards agent $j_{\lambda+1}$ for $\tilde{v}_1$, $\tilde{v}_1(W_1) \geq \rho \cdot v_{1,\lambda+1}$.
\end{itemize}
In all these cases, we have $\tilde{v}_1(W_1) \geq \rho \cdot v_{1,\lambda+1}$, and \eqref{eq:virtual:efx-approximation} gies us that the EFX approximation can be bounded as follows:
\begin{align*}
    \frac{v_1(X_1)}{v_1(X_i)} \geq \frac{ \rho \cdot v_{1,\lambda+1} }{ \theta_1^{-1} \cdot \bigg( v_{1,\lambda+1} + \min_{g \in X_i} \tilde{v}_1(g) \bigg) + \alpha_{i,k+1}\cdot v_{1,n-1} \theta_k }
\end{align*}
If $\min_{g \in X_i} \tilde{v}_1(g) > 0$, then $\alpha_{i,k+1} = 0$. Since $\theta_1^{-1} = m^{1/(k+1)}$ and $\min_{g \in X_i} \tilde{v}_1(g) \leq v_{1,\lambda+1}$, we have 
\begin{align*}
    \frac{v_1(X_1)}{v_1(X_i)} \geq \frac{ \rho \cdot v_{1,\lambda+1} }{ m^{1/(k+1)} \cdot \bigg( v_{1,\lambda+1} + v_{1,\lambda+1} \bigg)}
    = \frac{ \rho }{ 2 m^{1/(k+1)} }.
\end{align*}
Otherwise, by the induction hypothesis, the first top $\lambda$ goods of agent $1$ are assigned as singletons to $\lambda$ agents different than agent $1$ and agent $i$. Hence, $\min_{g \in X_i} \tilde{v}_1(g) \leq v_{1,\lambda+1}$. In addition, $\alpha_{i,k+1}\leq m$, $v_{1,n-1} \leq v_{1,\lambda+1}$ and $\theta_k = m^{-k/(k+1)}$. Hence, 
\begin{align*}
    \frac{v_1(X_1)}{v_1(X_i)} \geq \frac{ \rho \cdot v_{1,\lambda+1} }{ m^{1/(k+1)} \cdot v_{1,\lambda+1} + m\cdot v_{1,\lambda+1} m^{-k/(k+1)} }
    = \frac{ \rho }{ 2 m^{1/(k+1)} }.
\end{align*}
So, it must be the case that $X_{j_{\lambda+1}} = \{g_{1,\lambda+1}\}$ and the induction is complete. 

\medskip
\noindent 
By the above, we can assume that agents $1$ and $i$ are only given goods that agent $1$ ranks at positions in $\{n-1,\ldots,m\}$. 
We have the following two cases:
\begin{itemize}
    \item If $\min_{g \in X_i} \tilde{v}_1(g) > 0$, it must be the case that $\alpha_{i,k+1} = 0$. 
    Therefore, by \eqref{eq:virtual:efx-approximation}, the EFX approximation is
    \begin{align*}
    \frac{v_1(X_1)}{v_1(X_i)} \geq \frac{ \tilde{v}_1(X_1) }{ \theta_1^{-1} \cdot \bigg( \frac{1}{\rho} \tilde{v}_1(X_1) + \min_{g \in W_i} \tilde{v}_1(g) \bigg) }.
    \end{align*}    
    Since $X_i$ contains at least two goods, by the $\rho$-EFX property of $\tilde{v}_1$, we have that
    \begin{align*}
        \tilde{v}_1(X_1) 
        &\geq \rho \bigg(\tilde{v}_1(X_i) - \min_{g \in X_i} \tilde{v}_1(g) \bigg) \\
        &\geq \rho \cdot \min_{g \in X_i} \tilde{v}_1(g).
    \end{align*}
    Hence, 
    \begin{align*}
    \frac{v_1(X_1)}{v_1(X_i)} \geq 
    \frac{ \rho \cdot \min_{g \in X_i} \tilde{v}_1(g) }{ \theta_1^{-1} \cdot \bigg( \frac{1}{\rho} \rho \cdot \min_{g \in X_i} \tilde{v}_1(g) + \min_{g \in X_i} \tilde{v}_1(g) \bigg) } 
    = \frac{\rho}{2 m^{1/(k+1)}}.
    \end{align*}

    \item If $\min_{g \in X_i} \tilde{v}_1(g) = 0$,
    then, by \eqref{eq:virtual:efx-approximation}, the EFX approximation is
    \begin{align*}
    \frac{v_1(X_1)}{v_1(X_i)} \geq \frac{ \tilde{v}_1(X_1) }{ \theta_1^{-1} \cdot \frac{1}{\rho} \tilde{v}_1(X_1) + \alpha_{i,k+1}\cdot v_{1,n-1}\theta_k }.
    \end{align*}
    Observe now that it must be the case that $\tilde{v}_1(X_1) \geq \rho v_{1,n-1}$. 
    Indeed, as shown above, we have that $X_{j_\ell} = \{g_{1,\ell}\}$ for any $\ell \in [n-2]$. 
    Consequently, good $g_{1,n-1}$ is assigned to either agent $1$ or agent $i$. 
    In the former case, we have that $\tilde{v}_1(X_1) \geq v_{1,n-1}$. 
    In the latter case, due to the $\rho$-EFX property of $\tilde{v}_1$, we have that 
    $\tilde{v}_1(X_1) \geq \rho \bigg( v_{1,n-1} - \min_{g \in X_i} \tilde{v}_1(g) \bigg) = \rho v_{1,n-1}$.
    Using this, the definition of the thresholds, and the inequality $\alpha_{i,k+1} \leq m$, we obtain
    \begin{align*}
    \frac{v_1(X_1)}{v_1(X_i)} &\geq \frac{ \rho v_{1,n-1} }{ m^{1/(k+1)} \cdot \frac{1}{\rho}\rho v_{1,n-1} + m\cdot v_{1,n-1} m^{-k/(k+1)} } \\
    &= \frac{\rho}{2 m^{1/(k+1)}}.
    \end{align*}    
\end{itemize}
This completes the proof.
\end{proof}

Using the best-known $\rho$-EFX algorithm and Theorem~\ref{thm:virtual}, we obtain the following corollary, which shows that we can compute a constant-EFX allocation by making a rather small number of queries per agent.  

\begin{corollary} \label{cor:general-positive-constant}
With $O(n + \log^2{m})$ queries per agent, we can compute a $\frac{\phi-1}{2(1+\varepsilon)}$-EFX allocation, for any constant $\varepsilon>0$.
\end{corollary}

\begin{proof}
The best-known approximate-EFX algorithm that works for any number of agents and without any restrictions on the valuation functions is due to \citet{ANM2019} and achieves an EFX approximation factor $\rho = \phi-1\approx 0.618$. Let $\lambda$ be a constant and $k=\lambda \cdot \log{m}-1$. Then, Theorem~\ref{thm:virtual} implies that with $O(n + \log^2{m})$ queries per agent, we can compute an allocation that achieves an EFX approximation of $\frac{\rho}{2 m^{1/(\lambda\log{m})}} = \frac{\phi-1}{2 e^{1/\lambda}}$. For constant $\lambda$, since $\varepsilon$ is constant, we have $e^{1/\lambda} \leq 1+\varepsilon$. Therefore, the EFX approximation that we can achieve is at least $\frac{\phi-1}{2(1+\varepsilon)}$.
\end{proof}

We complement the above positive result with a negative one, which, as we will see later, implies (as a corollary) an impossibility on the necessary number of queries per agents to compute constant-EFX allocations. 

\begin{theorem}\label{thm:negative}
For any $k \geq 1$ such that $\sqrt{k} \leq m^{1/(2k-1)}$, no $k$-query algorithm can compute an allocation that is $\omega\left(\frac{\sqrt{k}}{m^{1/(2k-1)}}\right)$-EFX.
\end{theorem}

\begin{proof}
We consider instances in which the $n$ agents have the same ranking over the $m$ goods that is of the form 
$$g_1 \succ g_2 \succ ... \succ g_{n-1} \succ S_1 \succ \ldots \succ S_{k-1} \succ B.$$ 
In particular, for any $\ell \in [n-1]$, the agents rank the same good $g_\ell$ at position $\ell$. 
For any $\ell \in [k-1]$, set $S_\ell$ consists of $|S_\ell| = m^{(2\ell-1)/(2k-1)}$ goods. 
The last set $B$ consists of the remaining $|B| = m- n + 1 - \sum_{\ell=1}^{k-1} |S_\ell|$ goods; we choose $m$ large enough so that $|B| \geq \lambda m$ for some constant $\lambda \in (0,1)$. A query reveals value $\sqrt{k}$ for any good $g_\ell$ with $\ell \in [n-1]$, value $m^{-2\ell/(2k-1)}$ for any good $g\in S_\ell$ with $\ell \in [k-1]$, and value $0$ for any good $g \in B$. 
Let $\bX=(X_i)_{i \in [n]}$ be the allocation computed by the algorithm.

Observe that, since there are $n-1$ goods in the set $\{g_1, g_2, \ldots, g_{n-1}\}$ of top-ranked goods, there must be at least one agent that is not assigned any of these goods. Without loss of generality, let agent $n$ be such an agent. For any $g \in \{g_1, g_2, \ldots, g_{n-1}\}$ let $i_g$ be the agent that is assigned $g$ in $\bX$, that is, $g \in X_{i_g}$. We switch between the following two cases. 

\medskip
\noindent 
{\bf Case 1: There exists a good $g^* \in \{g_1, g_2, \ldots, g_{n-1}\}$ such that $|X_{i_{g^*}}|\geq 2$.}
Consider a valuation profile such that the valuation function of agent $n$ consists exactly of the values that would be revealed if the corresponding good were queried. That is,  
\begin{align*}
    v_n(g) = 
    \begin{cases}
        \sqrt{k},            & \text{if } g \in \{g_1, g_2, \ldots, g_{n-1}\} \\
        m^{-2\ell/(2k-1)},   & \text{if } g \in S_\ell, \ell \in [k-1] \\
        0,                   & \text{if } g \in B.  
    \end{cases}
\end{align*}
Then,
\begin{align*}
    v_n(X_n) 
    &\leq v_n(G\setminus\{g_1, g_2, \ldots, g_{n-1}\}) \\
    &= \sum_{\ell=1}^{k-1} |S_\ell| \cdot m^{-2\ell/(2k-1)} + |B|\cdot 0 \\
    &= \sum_{\ell=1}^{k-1} m^{(2\ell-1)/(2k-1)} \cdot m^{-2\ell/(2k-1)} \\
    &\leq k \cdot m^{-1/(2k-1)}.
\end{align*}
In addition, since $|X_{i_{g^*}}|\geq 2$, there exists a good $g \in X_{i_{g^*}}$ such that
\begin{align*}
    v_n(X_{i_{g^*}}\setminus\{g\}) \geq v_n(g^*) = \sqrt{k}. 
\end{align*}
These imply an EFX approximation of $O\left(\frac{\sqrt{k}}{m^{1/(2k-1)}}\right)$ for agent $n$ towards agent $i_{g^*}$.

\medskip
\noindent 
{\bf Case 2: $X_{i_{g}} = \{g\}$ for every $g \in \{g_1, g_2, \ldots, g_{n-1}\}$.}
We will argue that agent $n-1$ achieves the desired upper bound on the EFX approximation towards agent $n$; note that, in this case, since each agent in $[n-1]$ is assigned a singleton bundle, agent $n$ is assigned all goods in $G \setminus \{g_1, g_2, \ldots, g_{n-1}\}$. To simplify our notation, suppose that $X_{n-1}=\{g_{n-1}\}$. We switch between two cases depending on whether the value of agent $n-1$ for $g_{n-1}$ has been queried by the algorithm or not. 
\begin{itemize}[leftmargin=*]
    \item[] {\bf Agent $n-1$ is not queried for $g_{n-1}$.} 
    Since the value $v_{n-1}(g_{n-1})$ is unknown, we can set it to be as low as possible; in particular, we can set it to be equal to the value $m^{-2/(2k-1)}$ that could be revealed for the goods in set $S_1$ which are ranked below $g_{n-1}$. For any other good $g$, to account for all possible queries, we will set $v_{n-1}(g)$ to be the value that would be revealed if $g$ was queried. 
    That is, we have
    \begin{align*}
        v_{n-1}(g) = 
        \begin{cases}
            \sqrt{k},           & \text{if } g \in \{g_1, g_2, \ldots, g_{n-2}\} \\
            m^{-2/(2k-1)},       & \text{if } g = g_{n-1} \\
            m^{-2\ell/(2k-1)},   & \text{if } g \in S_\ell, \ell \in [k-1] \\
            0,                  & \text{if } g \in B \\
        \end{cases}
    \end{align*}
    Hence,
    \begin{align*}
        v_{n-1}(X_{n-1}) = v_{n-1}(g_{n-1}) = m^{-2/(2k-1)}
    \end{align*}
    and, for any $g \in B$, since $g \in X_n$,  
    \begin{align*}
        v_{n-1}(X_n \setminus \{g\}) 
        &= \sum_{\ell=1}^{k-1} v_{n-1}(g)\cdot m^{-2\ell/(2k-1)} + (|B|-1)\cdot 0 \\
        &= \sum_{\ell=1}^{k-1} m^{(2\ell-1)/(2k-1)} \cdot m^{-2\ell/(2k-1)} \\
        &= (k-1) \cdot m^{-1/(2k-1)}.
    \end{align*}
These imply an EFX approximation of $O\left(\frac{1}{k \cdot m^{1/(2k-1)}}\right)$ for agent $n-1$ towards agent $n$.

\item[] {\bf Agent $n-1$ is queried for $g_{n-1}$.}
Since there are $k-1$ queries left, there must be some set among the collection of $S_\ell$ for $\ell \in [k-1]$ and $B$ such that the algorithm does not learn the value of agent $n-1$ for any good therein. Whichever this set is, we can define the values of agent $n-1$ for those goods to be the upper bound (the revealed value for the goods in the set preceding it in the ranking), and all remaining values can be defined to be equal to the ones that would be revealed if queried by the algorithm. 

In particular, suppose agent $n-1$ is not queried for any good in set $S_\lambda$ for some $\lambda \in [k-1]$. We then consider the following consistent valuation function:
\begin{align*}
        v_{n-1}(g) = 
        \begin{cases}
            \sqrt{k},           & \text{if } g \in \{g_1, g_2, \ldots, g_{n-1}\} \\
            m^{-2\ell/(2k-1)},   & \text{if } g \in S_\ell, \ell \in [k-1]\setminus\{\lambda\} \\
            m^{-2(\lambda-1)/(2k-1)}, & \text{if } g \in S_\lambda \\
            0,                  & \text{if } g \in B \\
        \end{cases}
    \end{align*}
Hence,  
$$v_1(X_{n-1}) = v_{n-1}(g_{n-1}) = \sqrt{k},$$ 
and, for any $g \in B$, since $g \in X_n$, 
\begin{align*}
    v_{n-1}(X_{n-1}\setminus\{g\}) 
    &\geq |S_\lambda| \cdot m^{-2(\lambda-1)/(2k-1)} \\
    &= m^{(2\lambda-1)/(2k-1)} \cdot m^{-2(\lambda-1)/(2k-1)} \\
    &= m^{1/(2k-1)}.
\end{align*}
This implies an EFX approximation of $O\left(\frac{\sqrt{k}}{m^{1/(2k-1)}}\right)$ for agent $n-1$ towards agent $n$.

If agent $n-1$ is not queried for any good in set $B$, then we can define all values for goods not in $B$ to be the ones that would be revealed if queried, and $v_{n-1}(g)=m^{-2(k-1)/(2k-1)}$ for any $g \in B$. This again leads to the same asymptotic upper bound of $O\left(\frac{\sqrt{k}}{m^{1/(2k-1)}}\right)$ on the EFX approximation of agent $n-1$ towards agent $n$ since 
$v_{n-1}(X_{n-1})=\sqrt{k}$ 
and, for any $g \in B$,
$$v_{n-1}(X_n\setminus\{g\}) \geq (|B|-1) \cdot m^{-2(k-1)/(2k+1)} \geq \frac{\lambda}{2} m^{1/(2k+1)}.$$ 
\end{itemize}
The proof is now complete. 
\end{proof}

Using \cref{thm:negative}, we can show that any algorithm that makes $O\left(\frac{\log{m}}{\log{\log{m}}}\right)$ queries per agent does not compute a constant-EFX allocation; in other words, the allocations computed by such an algorithm achieve a superconstant EFX approximation. 

\begin{corollary} \label{cor:general-negative-contant}
For $k = O\left(\frac{\log{m}}{\log{\log{m}}}\right)$, no $k$-query algorithm can compute an allocation that achieves an EFX approximation of  
$\omega\left(\frac{1}{\sqrt{\log{\log{m}}}}\right)$.
\end{corollary}

\cref{cor:general-positive-constant} and \cref{cor:general-negative-contant} together narrow down the amount of queries per agent that are necessary to compute a constant-EFX allocation, but they also leave a gap between $O(n+\log^2{m})$ and $\Omega\left(\frac{\log{m}}{\log\log{m}}\right)$  which seems quite challenging to bridge. In the following sections, we show improved EFX approximation guarantees for two special cases, namely when the number $n$ of agents is constant, and for bivalued instances, i.e., when the valuation functions of the agents take two possible values.

\section{Constant Number of Agents} \label{sec:constant-n}
As we saw in the previous section, the {\sc $\rho$-Virtual-EFX} algorithm (\cref{alg:algorithm:virtual}) provides a variety of tradeoffs between the number of queries and EFX approximation, and, more importantly, it can be appropriately tuned to compute constant-EFX allocations. However, this algorithm also has a number of drawbacks.

First of all, it requires at least $n-1$ queries per agent to operate, and, of course, must make several more to achieve a good EFX approximation. So, this algorithm cannot even be employed in settings where the number of queries we are allowed to make per agent is fewer than the number of agents; note that such a scenario is quite realistic when responding to these queries comes with a high cognitive cost. In addition, even in instances where $n$ is a small constant (which is the most common case in practice, and also the ones we focus on in this section), it turns out that the algorithm does not compute allocations with the best possible EFX approximation as a function of the number $k$ of queries. To see this, let us restate \cref{thm:virtual} using constants $n$ and $\rho$, and transferring the $\log{m}$ factor from the upper bound on the number of queries to the lower bound on the EFX approximation. We have the following.

\begin{corollary} \label{cor:virtual-contant-n}
For constants $n$ and $\rho$, the {\sc $\rho$-Virtual-EFX} algorithm (\cref{alg:algorithm:virtual}) makes $O(k)$ queries per agent and computes an $\Omega\bigg( m^{-\frac{\log{m}}{k+\log{m}}} \bigg)$-EFX allocation.
\end{corollary}

\noindent 
We can now observe that the EFX approximation of the allocations computed by the algorithm cannot be larger than $1/\sqrt{m}$ for any number $k \leq \log{m}$ of queries per agent. In contrast, the negative result that we showed in the previous section (\cref{thm:negative}) suggests that it might be possible to compute allocations with an EFX approximation of $\Omega\left(\frac{1}{m^{1/(2k-1)}}\right)$ for constant $k$, and an approximation of $\Omega(1)$ for $k=O(\log{m})$. 

\bigskip

\noindent{\bf The \PRR algorithm.} In this section, we design a different algorithm that makes $k$ queries per agent, with $k$ that is always independent of $n$. We refer to this algorithm as \PRR (\sPRR), and present it in \cref{alg:algorithm:nagents-kqueries}. When $n$ is constant, the algorithm always computes an $\Omega\left( \frac{1}{\sqrt{k}\cdot m^{1/(2k-1)}}\right)$-EFX allocation. Observe that there is only a multiplicative-gap of $k$ between this positive result and the negative result shown in \cref{thm:negative}, and this gap vanishes when $k$ is constant, implying that the \sPRR algorithm achieves the best possible tradeoff for this special case. In addition, the algorithm computes $\Omega\left(1/\sqrt{\log{m}}\right)$-EFX allocations using $k = O(\log{m})$ queries, which is, however, its limit. 

The core idea of the \sPRR algorithm is the following: Agents with significantly higher value for their top-ranked good (among the currently available ones) compared to the other goods can be assigned just this good and be sufficiently satisfied at the end, whereas agents with relatively small value for their top-ranked good require more goods to avoid envying the other agents too much. To put this idea into practice, we aim to partition the agents into a set $B$ of agents that will be assigned a single good (that is their favorite available good at that point), and a set $R$ of agents that will be allocated the remaining goods by running the \RR algorithm. 

To determine these two sets, we start with all agents as {\em active} and part of $R$, and then run a {\em main loop} that repeats the following as long as there are active agents and some agent in $R$ (to maintain at least one agent that will be assigned more than one goods): 
Looking at the preference profile of the agents, we first identify the set $G_\succ$ all top-ranked goods (i.e., the goods that are ranked at the top by at least one agent). Afterwards, for each such good $g \in G_\succ$, we consider one by one, in an arbitrary ordering, the agents that rank $g$ at the top (set $P_g$). For each agent $i \in P_g$, we partition the goods from top to bottom into $k$ sets $S_{i,1},\ldots,S_{i,k}$ of increasing sizes which are determined by $k-1$ parameters $(\alpha_\ell)_{\ell \in [k-1]}$ that are given as input to the algorithm, and query the value $v_{i,\ell}^*$ of agent $i$ for the top-ranked good within $S_{i,\ell}$ for every $\ell \in [k]$.
Using $k-1$ parameters $(\beta_\ell)_{\ell \in [k-1]}$ that are given as input to the algorithm, we then check how large the value of the agent for her overall top-ranked good is compared to her value for each of the top-ranked goods in the other sets. In particular, we check whether $v_{i,1}^* \geq \beta_{\ell-1} v_{i,\ell}^*$ for every $\ell \in \{2,\ldots,k\}$. If this true, then agent $i$ is assigned $g$, $g$ is removed from the preference profile, $i$ is moved to set $B$ and becomes inactive, and the loop breaks to consider the next good in $G_\succ$ similarly. When the main loop terminates, we complete the allocation for the remaining agents in $R$ (which contains at least one agent) by running the \RR algorithm with input their preference profile for the unallocated goods (the ones not given to the agents of $B$ as singleton bundles). 

\begin{algorithm}[h!]
\caption{\textsc{\PRR} (\sPRR)}\label{alg:algorithm:nagents-kqueries}
\hspace*{\algorithmicindent} \textbf{Input:} Preference profile $\succ$, parameters $(\alpha_\ell)_{\ell \in [k-1]}$ and $(\beta_\ell)_{\ell \in [k-1]}$
\begin{algorithmic}[1]
\small
\State $N_A \gets N$ \Comment{set of active agents, initially all agents} 
\State $B \gets \varnothing$ \Comment{set of agent that get a singleton once they are queried, initially no one}
\State $R \gets N$ \Comment{set of agents that participate in the Round Robin phase, initially everyone}
\While{$N_A \neq \varnothing$ and $|B| < n-1$} \Comment{{\em Main loop}}
    \State $G_\succ \gets \varnothing$ \Comment{set of top-ranked goods in the current preference profile $\succ$, initially empty}
    \For{every $i\in N_A$}
        \State $G_\succ \gets G_\succ \cup \mathtt{top}_\succ(i)$ \Comment{Add the top-ranked good of $i$ in $G_\succ$}
    \EndFor
    \For{every $g \in G_\succ$}
        \State $P_g \gets \{i \in N_A: \mathtt{top}_\succ(i) = g \}$  \Comment{the set of agents that rank $g$ first in $\succ$}
    \EndFor

    \For{every $g \in G_\succ$}
        \For{every $i \in P_g$}
            
            \State Partition the available goods in $\succ$ into $k$ sets $S_{i,1}$, \ldots, $S_{i,k}$ from top to bottom such that: 
            $$\forall \ell \in [k-1], S_{i, \ell} \gets \text{the next $\alpha_\ell$ goods in $\succ_i$}$$ 
            $$S_{i,k} \gets \text{the last $m - \sum_{\ell=1}^{k-1}\alpha_\ell$ goods}$$

            \For{each $\ell \in [k]$} 
                \State Query agent $i$ for her top good $g^*_{i,\ell} \in S_{i,\ell}$ to learn value $v^*_{i,\ell}$
            \EndFor

            \State $N_A \gets N_A \setminus \{i\}$ \Comment{Agent $i$ has been queried and is no longer active}

            \If{$v^*_{i,1} \geq \beta_{\ell-1} \cdot v^*_{i,\ell}$ for all $\ell \in \{2,\ldots, k\}$} \Comment{If $i$ values highly her top-ranked good}
                \State $X_i \gets \{g^*_{i,1}\}$  \Comment{Agent $i$ is assigned only this good}
                \State $\succ \gets \succ \setminus \{g^*_{i,1}\}$ \Comment{and this good is removed from the preference profile.}
                \State $B \gets B \cup \{i\}$ \Comment{Agent $i$ is among the agents that get a singleton}
                \State $R \gets R \setminus \{i\}$ \Comment{and thus she will not participate in the Round Robin phase.}
                \State \textbf{break}
            \EndIf 
        \EndFor
    \EndFor
\EndWhile

\State Complete the allocation $\bX$ by running {\sc Round Robin} (\cref{alg:round-robin}) with input the agents of $R$ and their preference profile $\succ$ over the remaining goods. 

\State \Return $\bX=(X_i)_{i \in N}$
\end{algorithmic}
\end{algorithm}

By appropriately choosing the input parameters $(\alpha_\ell)_{\ell \in [k-1]}$ and $(\beta_\ell)_{\ell \in [k-1]}$, we can show that the \sPRR algorithm computes an allocation with a good EFX approximation guarantee, as follows. 

\begin{theorem} \label{thm:n-agents-k-queries}
For $k \geq 1$ and $\lambda \geq \max\left\{1,n/m^{1/(2k-1)}\right\}$, the \PRR algorithm (\cref{alg:algorithm:nagents-kqueries}) with input parameters $\alpha_\ell = \lambda\cdot  m^{(2\ell-1)/(2k-1)}$ and $\beta_\ell = \sqrt{k}\cdot m^{(2\ell)/(2k-1)}$ for any $\ell \in [k-1]$ computes an allocation that achieves an EFX approximation of 
$$\min\left\{\frac{1}{(\sqrt{k}+1)\cdot \lambda \cdot m^{1/(2k-1)}}, \frac{1}{1+\sqrt{k}\cdot \frac{n}{\lambda} \cdot m^{1/(2k-1)}}\right\}.$$
\end{theorem}

\begin{proof}
For the particular parameters, observe that 
\begin{align*}
&\alpha_1 = \lambda \cdot m^{1/(2k-1)}, \\
&\forall \ell \in \{2,\ldots,k-1\}:  \frac{\alpha_\ell}{\beta_{\ell-1}} = \frac{\lambda}{\sqrt{k}} \cdot m^{1/(2k-1)},\\
&\frac{m}{\beta_{k-1}} = \frac{m^{1/(2k-1)}}{\sqrt{k}},\\
&\forall \ell \in [k-1]: \frac{\beta_\ell}{\alpha_\ell} = \frac{\sqrt{k}}{\lambda} \cdot m^{1/(2k-1)}.
\end{align*}
We will argue about the EFX approximation that an agent $i$ achieves towards another agent $j$. Clearly, if $j \in B$, then $|X_j|=1$, and thus agent $i$ is EFX towards $j$. Hence, we can assume that $j \in R$. We now switch between two cases depending on whether $i \in B$ or $i \in R$. 

\medskip
\noindent 
{\bf Case 1: $i \in B$.}
We have that $X_i = \{g_{i,1}^*\}$ and $v_i(X_i) = v_{i,1}^* \geq \beta_{\ell-1} \cdot v^*_{i,\ell}$ for all $\ell \in \{2,\ldots, k\}$. Since agent $j$ is assigned goods during the \RR phase, we have that $X_j \subseteq \bigcup_{\ell=1}^k S_{i,\ell}$, and thus
\begin{align*}
    v_i(X_j) &\leq \sum_{\ell=1}^k \sum_{g \in S_{i,\ell}} v_i(g) \\
    &\leq \sum_{\ell=1}^k |S_{i,\ell}| \cdot v_{i,\ell}^* \\
    &\leq \bigg( \alpha_1 + \sum_{\ell=2}^{k-1} \frac{\alpha_\ell}{\beta_{\ell-1}} + \frac{m}{\beta_{k-1}} \bigg) \cdot v_{i,1}^* \\
    &\leq \bigg( \lambda m^{1/(2k-1)} + (k-2) \frac{\lambda}{\sqrt{k}} \cdot m^{1/(2k-1)} +  \frac{m^{1/(2k-1)}}{\sqrt{k}} \bigg) \cdot v_{i,1}^* \\
    &\leq (\sqrt{k}+1)\cdot \lambda \cdot m^{1/(2k-1)} \cdot v_{i,1}^*
\end{align*}
Hence, the EFX approximation of $i$ towards $j$ is 
\begin{align*}
    \frac{v_i(X_i)}{v_i(X_j)} \geq \frac{1}{(\sqrt{k}+1)\cdot \lambda \cdot m^{1/(2k-1)}}.
\end{align*}

\medskip
\noindent 
{\bf Case 2: $i \in R$.}
If agent $i$ picks before agent $j$ in the \RR phase, then, at the end, $i$ is envy-free towards agent $j$. So, we can assume that $i$ picks after $j$ in the \RR phase. Let $g_i^*$ be the top-ranked good of $i$ among all the goods that are available just before the \RR phase begins. Observe that all the goods that agent $i$ might rank higher than $g^*_{i,1} \in S_{i,1}$ are assigned as singletons even before $i$ was queried for $g_{i,1}^*$; thus, it must be the case that $i$ ranks $g^*_{i,1}$ higher than $g^*_i$, and we have that $v_i(g_i^*) \leq v_{i,1}^*$. Since the good that $i$ gets at each round of \RR is better than the one she loses to $j$ in the next round, we have that 
\begin{align*}
    v_i(X_j) \leq v_i(X_i) + v_i(g_i^*) \leq v_i(X_i) + v_{i,1}^*.
\end{align*}
The fact $i \in R$ implies that there is an $\ell \in \{2,\ldots,k\}$ such that $v_{i,1}^* < \beta_{\ell-1} \cdot v^*_{i,\ell}$. Out of the $\alpha_{\ell-1}$ goods in $S_{i,\ell-1}$, at most $|B| = n -|R|$ of them might be allocated as singletons to the agents of $B$, while the remaining $\alpha_{\ell-1}-n+|R|$ goods are allocated during the \RR phase. Consequently, $i$ gets a set of goods that she overall values at least as much as a $1/|R|$-fraction of these $\alpha_{\ell-1}-n+|R|$ goods, each of which is valued as at least $v_{i,\ell}^*$. Hence, we have that
\begin{align*}
    v_i(X_i) \geq \frac{\alpha_{\ell-1}-n+|R|}{|R|} \cdot v_{i,\ell}^* \geq \frac{\alpha_{\ell-1}}{n} \cdot v_{i,\ell}^*, 
\end{align*}
where the last inequality follows since the expression $\frac{\alpha_{\ell-1}-n}{|R|}+1$ is decreasing in $|R| \leq n$; note that $a_{\ell-1} > n$ for any $\lambda > \max\left\{1,\frac{n}{m^{1/(2k-1)}}\right\}$.

Using all of the above, the EFX approximation of $i$ towards $j$ is
\begin{align*}
    \frac{v_i(X_i)}{v_i(X_j)} \geq \frac{v_i(X_i)}{v_i(X_i) + v_{i,1}^*} \geq \frac{v_i(X_i)}{v_i(X_i) + \beta_{\ell-1}\cdot v^*_{i,\ell}}.
\end{align*}
The last expression is increasing in $v_i(X_i)$, and thus
\begin{align*}
    \frac{v_i(X_i)}{v_i(X_j)} 
    \geq \frac{\frac{\alpha_{\ell-1}}{n} \cdot v_{i,\ell}^*}{\frac{\alpha_{\ell-1}}{n} \cdot v_{i,\ell}^* + \beta_{\ell-1}\cdot v^*_{i,\ell}}  
    = \frac{1}{1 + n \cdot \frac{\beta_{\ell-1}}{\alpha_{\ell-1}}} 
    = \frac{1}{1 + \sqrt{k} \cdot \frac{n}{\lambda} \cdot m^{1/(2k-1)}}.
\end{align*}
This completes the proof.
\end{proof}

By \cref{thm:n-agents-k-queries} and \cref{thm:negative}, we obtain the following corollaries.

\begin{corollary}\label{cor:constant-n-general-k}
For constant number $n$ agents, the best possible EFX approximation is $\Theta\left(\frac{1}{\sqrt{k} \cdot m^{1/(2k-1)}}\right)$.
\end{corollary}

\begin{proof}
 Since $n$ is constant, the $\lambda$ parameter in \cref{thm:n-agents-k-queries} is a constant, and thus the \sPRR algorithm computes an $\Omega\left(\frac{1}{\sqrt{k}\cdot m^{1/(2k-1)}}\right)$-EFX allocation.
\end{proof}

\begin{corollary} \label{cor:constant-n-constant-k}
For constant number $n$ agents and constant number $k$ of queries, the best possible EFX approximation is $\Theta\left(\frac{1}{m^{1/(2k-1)}}\right)$.
\end{corollary}

\begin{proof}
 Since $n$ is constant, the $\lambda$ parameter in \cref{thm:n-agents-k-queries} is a constant, and thus the \sPRR algorithm computes an $\Omega\left(\frac{1}{m^{1/(2k-1)}}\right)$-EFX allocation. The upper bound follows by \cref{thm:negative} for constant $k$.
\end{proof}

\begin{corollary} \label{cor:constant-n-k=logm}
For a constant number $n$ agents, there is an algorithm that makes $O(\log{m})$ queries per agent and computes an $\Omega(1/\sqrt{\log{m}})$-EFX allocation.
\end{corollary}

\begin{proof}
Since $n$ is constant and $m^{1/(2k-1)}$ is constant for $k = O(\log{m})$, the $\lambda$ parameter in \cref{thm:n-agents-k-queries} is also a constant, and the \sPRR algorithm computes an $\Omega(1/\sqrt{\log{m}})$-EFX allocation.
\end{proof}

\section{Bivalued Instances}\label{sec:bivalued}
In this section, we consider a class of well-structured instances in which the agents are known to have only two possible (personalized) values for the goods. In particular, in a {\em bivalued} instance, for every agent $i$, $v_i(g) \in \{h_i,\ell_i\}$ for any good $g \in G$, where $h_i$ and $\ell_i$ are the two possible (real, non-negative) values that agent $i$ assigns to the goods such that $h_i > \ell_i$. We will refer to $h_i$ as the {\em high} value of agent $i$, and, similarly, we will refer to $\ell_i$ as the {\em low} value of agent $i$.

For bivalued instances, when the valuation functions are fully known, EFX allocations always exist and can be computed efficiently in polynomial time by (variants of) the {\sc Match$\&$Freeze} algorithm which repeatedly computes maximum matchings and freezes certain agents to eliminate excessive envy. 
The ``vanilla'' {\sc Match$\&$Freeze} algorithm of \citet{amanatidis2021maximum} was originally designed to work for bivalued instances in which $h_i=h$ and $\ell_i=\ell$ for every agent $i$, but recently \citet{jin2025pareto} and \citet{byrka2025probing} presented modified versions of the algorithm that compute EFX allocations even when the agents have personalized values. 

In more detail, this algorithm works in rounds, as follows:  
In each round, it computes a maximum matching between the {\em unfrozen} agents (initially all agents) and the goods they value as high.
If the matching is perfect (i.e., each agent is matched to a high-valued good), then the agents are assigned the goods they are matched to. Otherwise, if there is no perfect maximum matching, the agents with higher ratio of high to low value are given priority and are assigned the goods they are matched to in a max-cardinality maximum matching. 
Any agent $i$ among the remaining ones (that are not matched to high-valued goods) is arbitrary assigned a low-valued good. 
In case agent $i$ had high value for goods that were, however, assigned to other agents due to the computed maximum matching, then all these agents {\em freeze} for a certain number of rounds to avoid excessive envy between $i$ and these agents; to be exact, these agents freeze for $\lfloor h_i/\ell_i -1 \rfloor$ rounds so that $i$ can catch up to the high value she lost using only low-valued goods. 
We present the pseudocode of the algorithm in \cref{alg:match&freeze}.  

\begin{algorithm}[!h]
\caption{\textsc{Match$\&$Freeze}}\label{alg:match&freeze}
\begin{algorithmic}[1]
\small
\For{each agent $i$}
    \State $X_i \gets \varnothing$ \Comment{Empty initial allocation}
    \State $F_i \gets 0$        \Comment{{\em freeze counter} of agent $i$, initially $0$.}
\EndFor
\State $\calP\gets G$ \Comment{The pool of available goods, initially all are available}
\While{$\calP \neq \varnothing$}
\State $U \gets \varnothing$ \Comment{The set of unfrozen agents in this round that will be assigned a good.}
\For{each agent $i \in N$}
    \If{$F_i > 0$}
        \State $F_i \gets F_i -1$ \Comment{If the agent is frozen in this round, reduce freeze counter by one.}
    \Else
        \State $U \gets U \cup \{i\}$ \Comment{If the agent is not frozen, add her to set $U$.}
    \EndIf
\EndFor
\State $\mu \gets $ (prioritized) maximum matching between agents in $U$ and their high-valued goods.
\For{each agent $i \in \mu$}
    \State $X_i \gets X_i \cup \{\mu(i)\}$  \Comment{$i$ is assigned the good $\mu(i)$ she is matched to in $\mu$}
    \State $\calP \gets \calP \setminus \{\mu(i)\}$ \Comment{and this good is no longer available.}
\EndFor
\For{each agent $i \not\in \mu$}
\If{$\calP \neq \varnothing$}
   \State $X_i \gets X_i \cup \{$ arbitrary good $g$ in $\calP\}$ \Comment{$i$ is assigned an arbitrary available good}
   \State $\calP \gets \calP \setminus \{g\}$ \Comment{and this good is no longer available.}
   \For {each agent $j \in A\cap\mu$}
        \If {$v_i(\mu(j))=h_i$}             \Comment{If $j$ ``stole'' a high-value good from $i$, $j$ must freeze}
            \State $F_j \gets \lfloor h_i/\ell_i -1 \rfloor$  \Comment{for this many rounds.}
        \EndIf
   \EndFor
\EndIf
\EndFor
\EndWhile
\State \Return $\bX = (X_1,\ldots,X_n)$
\end{algorithmic}
\end{algorithm}

In our ordinal model, it is not hard to observe that it is possible to compute an exact-EFX allocation with $O(\log{m})$ queries per agent. Specifically, for each agent, we may perform binary search over her preference ranking (of size $m$) to find the point of transition (if any) from a high value to a low value, and thus fully uncover the valuation profile with $O(\log{m})$ queries. Once we have all this information, we can simply employ {\sc Match$\&$Freeze} (\cref{alg:match&freeze}) to output an EFX allocation. So, our goal for bivalued instances is not to just achieve a good approximation of EFX, but to do so with significantly fewer queries. We first show that we can always compute a $1/2$-EFX allocation using $O(\log{n})$ queries; note that the number $n$ of agents is significantly smaller than the number $m$ of goods in most applications, and thus this bound is significantly improved over the aforementioned $O(\log m)$ bound.  

Our algorithm, coined \MFRR (\sMFRR), works as follows: 
It first runs a binary search on the top-$n$ ranked goods of each agent looking for a threshold good therein such that the values of the agent drop from high to low. If this is the case, then, since there are only two possible values, this is sufficient to uncover the whole valuation function of the agent, and such an agent becomes part of a set $M$. On the other hand, if no such good exists in the top-$n$ ones, then we have an agent with the same (high or low) value for all top-$n$ ranked goods, and such an agent becomes part of a set $R$. 
The algorithm then allocates the goods in phases such that each agent is allocated at most one good during each phase. 
During a phase, the algorithm first assigns goods to the agents of $M$ by running a single round of {\sc Match$\&$Freeze} (\cref{alg:match&freeze}) as a subroutine, and, afterwards, it assigns goods to the agents of $R$ by running a single round of {\sc Round Robin} (\cref{alg:round-robin}). See \cref{alg:MF&RR:short} for a short description of the \sMFRR algorithm with pseudocode, and \cref{alg:MF&RR:full} in the appendix for the full version (which more formally describes how a round of {\sc Match$\&$Freeze} and a round of {\sc Round Robin} run the one after the other during the same phase).

\begin{algorithm}[!ht]
\caption{\MFRR (short version)}\label{alg:MF&RR:short}
\begin{algorithmic}[1]
\small
\State $M \gets \varnothing$  \Comment{Agents with known valuation function that contains a transition point.}
\State $R \gets \varnothing$  \Comment{Agents with possibly unknown valuation function with same value for all top-$n$ goods.}
\For{each agent $i \in N$} 
    \State Binary search over top-$n$ of $\succ_i$ for two consecutive goods $g$ and $g'$ such that $v_i(g) > v_i(g')$.
    \If{the search is successful}
        \State $M \gets M \cup \{i\}$ \Comment{$i$ becomes part of $M$}
    \Else
        \State $R \gets R \cup \{i\}$ \Comment{$i$ becomes part of $R$}
    \EndIf
\EndFor
\While{there are available goods}
\State Run a round of {\sc Match$\&$Freeze} (\cref{alg:match&freeze}) for the agents of $M$
\State Run a round of {\sc Round Robin} (\cref{alg:round-robin}) for the agents of $R$
\EndWhile
\State \Return $\bX = (X_1,\ldots,X_n)$
\end{algorithmic}
\end{algorithm}

\begin{theorem} \label{thm:bivalued-logn}
    For bivalued instances, the \MFRR algorithm (\cref{alg:MF&RR:short}) makes $O(\log{n})$ queries per agent and computes a $1/2$-EFX allocation. 
\end{theorem}

\begin{proof}
    First observe that the algorithm indeed makes $O(\log{n})$ queries per agent since it runs a binary search over the top-$n$ goods of each agent looking for a threshold good (if it exists).    
    By the definition of the algorithm, in each phase, the agents in $M$ who are assigned goods according to {\sc Match$\&$Freeze} have priority over the agents in $R$ who are assigned goods according to \RR. In addition, at the end of each phase, every unfrozen agent  is assigned one good. Finally, observe that an agent in $M$ can only freeze at most once during the execution of the algorithm; this is true since, after freezing, an agent has only low values for the remaining goods (note that if this was not the case, a larger matching would exist just before the agent got frozen, and thus the agent would not get frozen).
    We now argue that any agent $i$ is at least $1/2$-EFX towards any other agent $j$. 

\medskip
\noindent 
{\bf Case 1: $i \in M$ and $j \in M$.}
Since both $i$ and $j$ are assigned goods according to {\sc Match$\&$Freeze}, which is known to compute EFX allocations~\citep{amanatidis2021maximum,jin2025pareto,byrka2025probing}, we have that $i$ is EFX towards $j$. 

\medskip
\noindent 
{\bf Case 2: $i \in M$ and $j \in R$.}
First consider the case where agent $i$ never gets frozen during the execution of  the algorithm. In any phase of the algorithm, agent $i$ is assigned a good before agent $j$ since the agents of $M$ have priority over the agents of $R$. Since this single-phase allocation is done by computing a maximum matching among the agents in $M$ and their high-valued goods, agent $i$ is assigned a low-valued good only if all her high-valued goods have been depleted. Consequently, agent $i$ is always assigned a good that she values at least as much as the good assigned to $j$ in every phase, and thus $i$ is envy-free towards $j$ at the end. 

Next, consider the case where agent $i$ becomes frozen during the $k$-th phase of the algorithm and remains frozen for a number of phases. 
This happens because {\sc Match$\&$Freeze} assigns agent $i$ a good that both she and another agent $t$ value as high, and agent $t$ is assigned a low-valued good. Since this allocation is done via a maximum matching, there is no matching such that both $i$ and $t$ are assigned high-value goods, and thus $i$ must have low value for all remaining goods (including the one that agent $j$ is assigned in the $k$-th phase). 
Since agent $i$ is EFX towards agent $t$ at the end of the algorithm, $t$ is assigned goods before agent $j$, and $i$ has value $\ell_i$ for all goods that remain once she gets frozen, it must be the case that $i$ is also EFX towards agent $j$ at the end of the algorithm. 

\medskip
\noindent 
{\bf Case 3: $i \in R$ and $j \in R$ or $j \in M$.}
If agent $i$ is assigned goods before agent $j$ (which can only happen when $j \in R$), then $i$ ends up being envy-free towards $j$. 
So, we can assume that $i$ is assigned goods after $j$. 
Let $v \in \{h_i,\ell_i\}$ be the value that agent $i$ has for all $n$ goods within her top-$n$, and observe that $i$ is assigned at least one of these goods, which implies that $v_i(X_i) \geq v$. 
Since $i$ values the good she is assigned in each phase of the algorithm at least as much as the good that $j$ is assigned in the next phase, we have that $v_i(X_j) \leq v_i(X_i) + v$. 
Therefore, the EFX-approximation of $i$ towards  $j$ is
\begin{align*}
    \frac{v_i(X_i)}{v_i(X_j)} \geq \frac{v_i(X_i)}{v_i(X_i) + v} \geq \frac12.
\end{align*}
We conclude that any agent is at least $1/2$-EFX towards any other agent, and the computed allocation is $1/2$-EFX, as desired.
\end{proof}

Quite interestingly, we can also show that two queries per agent are sufficient to compute $1/n$-EFX allocations for bivalued instances; Note that the $1/n$ bound is independent of the number $m$ of agents, which is typically the larger parameter, and is constant in most settings where the number of agents is small. This is achieved by running the \PRR algorithm (\cref{alg:algorithm:nagents-kqueries}) of the previous section for $k=2$ and appropriately defined parameters $\alpha_1$ and $\beta_1$. The fact that there are two possible values per agent allows us to perform a tighter analysis on the EFX approximation bound that the computed allocation achieves (compared to \cref{thm:n-agents-k-queries}), and obtain the following result.   

\begin{theorem}
    For bivalued instances, there is an algorithm that makes $2$ queries per agent and computes a $1/n$-EFX allocation. 
\end{theorem}

\begin{proof}
    We consider the \PRR algorithm (\cref{alg:algorithm:nagents-kqueries}) presented in \cref{sec:constant-n} for $k=2$ with input parameters $\alpha = \alpha_1 = n-1$ and $\beta = \beta_1 = m/2$. 
    We now argue about the EFX approximation that an agent $i$ achieved towards another agent $j$, similarly to the proof of \cref{thm:n-agents-k-queries}. If $j \in B$, then $|X_j|=1$, and thus agent $i$ is EFX towards agent $j$. So, we focus on the case $j \in R$. 

\medskip
\noindent
{\bf Case 1: $i \in B$.}
By definition, for the particular parameters $\alpha$ and $\beta$, when the algorithm considers agent $i$, it queries $i$'s value for her currently top-ranked good $g_{i,1}$ and another good $g_{i,n}$ that appears $n$ positions below in her ranking. 
Since $i\in B$, we have that $v(g_{i,1}) \geq \beta \cdot v(g_{i,n})$, and $X_i = \{g_{i,1}\}$. For the inequality to be true, since agent $i$ has two possible values, $h_i$ and $\ell_i$, it must be the case that $v(g_{i,1})=h_i$ and $v(g_{i,n})=\ell_i$. Since $j \in R$, $X_j$ is a subset of all the goods allocated after $i$ is assigned $X_i$. In the ranking of $i$, there are $n-2$ goods strictly between $g_{i,1}$ and $g_{i,n}$ for which $i$ might have value $h_i$, and thus
\begin{align*}
    v_i(X_j) \leq (n-2) \cdot h_i + (m-n+2) \cdot h_\ell \leq \bigg( n-2 + 2\cdot \frac{m-n+2}{m} \bigg) h_i \leq n \cdot h_i.
\end{align*}
Hence, the EFX approximation of $i$ towards $j$ is $1/n$. 

\medskip
\noindent
{\bf Case 2: $i \in R$.} 
If agent $i$ picks before agent $j$ in the \RR phase of the algorithm, then, at the end, $i$ is envy-free towards agent $j$. So, we can assume that $i$ picks after $j$. Observe that agent $i$ has been queried for two goods $g_{i,1}$ and $g_{i,n}$ that are $n$ positions apart in her ranking. There are two possibilities for the corresponding values:
(a) $v_i(g_{i,1}) = v_i(g_{i,n}) \in \{h_i,\ell_i\}$, or 
(b) $v_i(g_{i,1})=h_i$, $v_i(g_{i,n})=\ell_i$, and, since $i \in R$, $h_i < \beta \cdot \ell_i$. 
For the good $g_i^*$ that $i$ ranks at the top and is available just before the \RR phase starts, we have that 
$v_i(g_i^*) \leq v_i(g_{i,1})$. Since the good that $i$ gets in each round is better than the one she loses to $j$ in the next round, we have that 
\begin{align*}
    v_i(X_j) \leq v_i(X_i) + v_i(g_i^*) \leq v_i(X_i) + v_i(g_{i,1}).
\end{align*}
The EFX approximation of $i$ towards $j$ is then
\begin{align*}
    \frac{v_i(X_i)}{v_i(X_j)} \geq \frac{v_i(X_i)}{v_i(X_i) + v_i(g_{i,1})}.
\end{align*}
If $v_i(g_{i,1}) = v_i(g_{i,n})$, then since agent $i$ is guaranteed to get at least one of her top-$n$ ranked goods in the \RR phase, we have that $v_i(X_i) \geq v_i(g_{i,1})$, and the EFX approximation is at least $1/2$. Otherwise, if $v_i(g_{i,1})=h_i$, $v_i(g_{i,n})=\ell_i$ and $h_i < \beta \cdot \ell_i$, we have that $v_i(X_i) \geq \frac{m}{n} \cdot \ell_i$ since agent $i$ gets at least one low-valued good in every round of the \RR phase, and thus 
\begin{align*}
    \frac{v_i(X_i)}{v_i(X_j)} \geq \frac{ \frac{m}{n} \cdot h_i}{ \frac{m}{n} \cdot \ell_i + h_i} \geq \frac{1}{1 + \frac{n}{m}\beta} = \frac{1}{1 + n/2} \geq \frac{1}{n}.
\end{align*}
This completes the proof.
\end{proof}

\section{Future Directions} \label{sec:conclusion}

Our work leaves some challenging gaps and also opens up interesting directions for future work along different dimensions. 
In terms of our results, the most important question is to fully characterize the Pareto frontier of the number of queries required to achieve constant-EFX, either asymptotically (along the lines of our investigation in this paper) or, in a more refined manner, aiming to optimize the specific constants also. Of course, as demonstrated in our work, this is a quite demanding task, and thus a more tangible goal would be to focus on special cases, including those of constant $n$ and bivalued instances that we considered, but also other important ones, such as $3$-valued instances~\citep{amanatidis2024pushing}, or instances in which the values of the agents can be encoded using (multi)graphs~\citep{christodoulou2023fair,amanatidis2024pushing,sgouritsa2025existence}.

Besides the $\alpha$-EFX approximations that we focused on in this work, our ordinal-information model can be applied to other (more relaxed) types of EFX approximation, such as {\em $\alpha$-EFX with charity}, where some goods can be left unallocated as long as the agents do not (approximately) envy them~\citep{caragiannis2019charity,CKMS20,akrami2025efx,hv2025almost}. In addition, one can also consider alternative fairness notions, such as relaxations of envy-freeness that are ``between'' EFX and EF1, or relaxations of proportionality and MMS, and further combine those with efficiency measures such as social welfare or Pareto optimality. Orthogonal to that, it would be interesting to understand the (im)possibilities of achieving good approximations to EFX (and other fairness notions) in settings where the resources to allocate are not goods, but {\em chores}, for which the agents have costs rather than values~\citep{akrami2023fair,sun2023fairness,garg2025constant,lin2025approximately}. 

The value query model we have considered here dependents on the assumption that the underlying valuation functions of the agents are additive, which allows us to focus on eliciting values for individual goods only. However, in many real-life applications, goods may be complements or substitutes, and hence different combinations of goods might be more or less valuable than the sum of them individually. The fair division literature has also extensively studied such settings with more complex valuations (e.g., submodular or subadditive valuations) in a regime of full information. Naturally, it makes perfect sense to also consider such models when only limited information about the valuations is available. However, for such valuations, we may need to assume access to preference rankings over (some of the possible) bundles and be allowed to query the values of the agents for whole bundles, rather than separate goods. It will be quite interesting to understand the minimum information requirements to achieve good EFX approximations in these settings.

\bibliographystyle{plainnat}
\bibliography{references}

\begin{thebibliography}{42}
\providecommand{\natexlab}[1]{#1}
\providecommand{\url}[1]{\texttt{#1}}
\expandafter\ifx\csname urlstyle\endcsname\relax
  \providecommand{\doi}[1]{doi: #1}\else
  \providecommand{\doi}{doi: \begingroup \urlstyle{rm}\Url}\fi

\bibitem[Akrami et~al.(2023)Akrami, Chaudhury, Garg, Mehlhorn, and Mehta]{akrami2023fair}
Hannaneh Akrami, Bhaskar~Ray Chaudhury, Jugal Garg, Kurt Mehlhorn, and Ruta Mehta.
\newblock {Fair and Efficient Allocation of Indivisible Chores with Surplus}.
\newblock In \emph{Proceedings of the 32nd International Joint Conference on Artificial Intelligence ({IJCAI})}, pages 2494--2502, 2023.

\bibitem[Akrami et~al.(2025)Akrami, Alon, Chaudhury, Garg, Mehlhorn, and Mehta]{akrami2025efx}
Hannaneh Akrami, Noga Alon, Bhaskar~Ray Chaudhury, Jugal Garg, Kurt Mehlhorn, and Ruta Mehta.
\newblock {EFX: A Simpler Approach and an (Almost) Optimal Guarantee via Rainbow Cycle Number}.
\newblock \emph{Operations Research}, 73\penalty0 (2):\penalty0 738--751, 2025.

\bibitem[Amanatidis et~al.(2020)Amanatidis, Markakis, and Ntokos]{ANM2019}
Georgios Amanatidis, Evangelos Markakis, and Apostolos Ntokos.
\newblock {Multiple birds with one stone: Beating 1/2 for {EFX} and {GMMS} via envy cycle elimination}.
\newblock \emph{Theoretical Computer Science}, 841:\penalty0 94--109, 2020.

\bibitem[Amanatidis et~al.(2021{\natexlab{a}})Amanatidis, Birmpas, Filos-Ratsikas, Hollender, and Voudouris]{amanatidis2021maximum}
Georgios Amanatidis, Georgios Birmpas, Aris Filos-Ratsikas, Alexandros Hollender, and Alexandros~A Voudouris.
\newblock {Maximum Nash welfare and other stories about EFX}.
\newblock \emph{Theoretical Computer Science}, 863:\penalty0 69--85, 2021{\natexlab{a}}.

\bibitem[Amanatidis et~al.(2021{\natexlab{b}})Amanatidis, Birmpas, Filos-Ratsikas, and Voudouris]{amanatidis2021peeking}
Georgios Amanatidis, Georgios Birmpas, Aris Filos-Ratsikas, and Alexandros~A Voudouris.
\newblock {Peeking behind the ordinal curtain: Improving distortion via cardinal queries}.
\newblock \emph{Artificial Intelligence}, 296:\penalty0 103488, 2021{\natexlab{b}}.

\bibitem[Amanatidis et~al.(2022)Amanatidis, Birmpas, Filos{-}Ratsikas, and Voudouris]{amanatidis2022matching}
Georgios Amanatidis, Georgios Birmpas, Aris Filos{-}Ratsikas, and Alexandros~A. Voudouris.
\newblock {A Few Queries Go a Long Way: Information-Distortion Tradeoffs in Matching}.
\newblock \emph{Journal of Artificial Intelligence Research}, 74, 2022.

\bibitem[Amanatidis et~al.(2023)Amanatidis, Aziz, Birmpas, Filos-Ratsikas, Li, Moulin, Voudouris, and Wu]{amanatidis2023fair}
Georgios Amanatidis, Haris Aziz, Georgios Birmpas, Aris Filos-Ratsikas, Bo~Li, Herv{\'e} Moulin, Alexandros~A Voudouris, and Xiaowei Wu.
\newblock Fair division of indivisible goods: Recent progress and open questions.
\newblock \emph{Artificial Intelligence}, 322:\penalty0 103965, 2023.

\bibitem[Amanatidis et~al.(2024{\natexlab{a}})Amanatidis, Birmpas, Filos{-}Ratsikas, and Voudouris]{amanatidis2024dice}
Georgios Amanatidis, Georgios Birmpas, Aris Filos{-}Ratsikas, and Alexandros~A. Voudouris.
\newblock {Don't Roll the Dice, Ask Twice: The Two-Query Distortion of Matching Problems and Beyond}.
\newblock \emph{{SIAM} Journal on Discrete Mathematics}, 38\penalty0 (1):\penalty0 1007--1029, 2024{\natexlab{a}}.

\bibitem[Amanatidis et~al.(2024{\natexlab{b}})Amanatidis, Filos-Ratsikas, and Sgouritsa]{amanatidis2024pushing}
Georgios Amanatidis, Aris Filos-Ratsikas, and Alkmini Sgouritsa.
\newblock {Pushing the frontier on approximate EFX allocations}.
\newblock In \emph{Proceedings of the 25th {ACM} Conference on Economics and Computation ({EC})}, pages 1268--1286, 2024{\natexlab{b}}.

\bibitem[Anshelevich et~al.(2021)Anshelevich, Filos-Ratsikas, Shah, and Voudouris]{distortion-survey}
Elliot Anshelevich, Aris Filos-Ratsikas, Nisarg Shah, and Alexandros~A. Voudouris.
\newblock {Distortion in Social Choice Problems: The First 15 Years and Beyond}.
\newblock In \emph{Proceedings of the 30th International Joint Conference on Artificial Intelligence {(IJCAI)}}, pages 4294--4301, 2021.

\bibitem[Babaioff et~al.(2021)Babaioff, Ezra, and Feige]{babaioff2021fair}
Moshe Babaioff, Tomer Ezra, and Uriel Feige.
\newblock {Fair and Truthful Mechanisms for Dichotomous Valuations}.
\newblock In \emph{Proceedings of the 35th {AAAI} Conference on Artificial Intelligence ({AAAI})}, pages 5119--5126, 2021.

\bibitem[Benad{\`{e}} et~al.(2022)Benad{\`{e}}, Halpern, and Psomas]{benade2022dynamic}
Gerdus Benad{\`{e}}, Daniel Halpern, and Alexandros Psomas.
\newblock {Dynamic Fair Division with Partial Information}.
\newblock In \emph{Proceedings of the 36th Annual Conference on Neural Information Processing Systems ({NeurIPS})}, 2022.

\bibitem[Bouveret et~al.(2010)Bouveret, Endriss, and Lang]{bouveret2010fair}
Sylvain Bouveret, Ulle Endriss, and J{\'{e}}r{\^{o}}me Lang.
\newblock {Fair Division under Ordinal Preferences: Computing Envy-Free Allocations of Indivisible Goods}.
\newblock In \emph{Proceedings of the 19th European Conference on Artificial Intelligence ({ECAI})}, pages 387--392. 2010.

\bibitem[Bu et~al.(2024)Bu, Li, Liu, Song, and Tao]{bu2024logarithmic}
Xiaolin Bu, Zihao Li, Shengxin Liu, Jiaxin Song, and Biaoshuai Tao.
\newblock Logarithmic comparison-based query complexity for fair division of indivisible goods.
\newblock In \emph{Proceedings of the 20th International Conference on Web and Internet Economics ({WINE})}, pages 348--365, 2024.

\bibitem[Budish(2011)]{budish2011combinatorial}
Eric Budish.
\newblock {The Combinatorial Assignment Problem: Approximate Competitive Equilibrium from Equal Incomes}.
\newblock \emph{Journal of Political Economy}, 119\penalty0 (6):\penalty0 1061--1103, 2011.

\bibitem[Byrka et~al.(2025)Byrka, Malinka, and Ponitka]{byrka2025probing}
Jaroslaw Byrka, Franciszek Malinka, and Tomasz Ponitka.
\newblock {Probing {EFX} via {PMMS:} (Non-)Existence Results in Discrete Fair Division}.
\newblock \emph{CoRR}, abs/2507.14957, 2025.

\bibitem[Caragiannis and Fehrs(2024)]{caragiannis2024beyond}
Ioannis Caragiannis and Karl Fehrs.
\newblock {Beyond the Worst Case: Distortion in Impartial Culture Electorates}.
\newblock In \emph{Proceedings of the 20th International Conference on Web and Internet Economics ({WINE})}, pages 420--437, 2024.

\bibitem[Caragiannis et~al.(2019{\natexlab{a}})Caragiannis, Gravin, and Huang]{caragiannis2019charity}
Ioannis Caragiannis, Nick Gravin, and Xin Huang.
\newblock Envy-freeness up to any item with high {N}ash welfare: The virtue of donating items.
\newblock In \emph{Proceedings of the 2019 {ACM} Conference on Economics and Computation ({EC})}, pages 527--545, 2019{\natexlab{a}}.

\bibitem[Caragiannis et~al.(2019{\natexlab{b}})Caragiannis, Kurokawa, Moulin, Procaccia, Shah, and Wang]{caragiannis2019unreasonable}
Ioannis Caragiannis, David Kurokawa, Herv{\'e} Moulin, Ariel~D Procaccia, Nisarg Shah, and Junxing Wang.
\newblock {The Unreasonable Fairness of Maximum Nash Welfare}.
\newblock \emph{ACM Transactions on Economics and Computation}, 7\penalty0 (3):\penalty0 12:1--12:32, 2019{\natexlab{b}}.

\bibitem[Chaudhury et~al.(2021)Chaudhury, Kavitha, Mehlhorn, and Sgouritsa]{CKMS20}
Bhaskar~Ray Chaudhury, Telikepalli Kavitha, Kurt Mehlhorn, and Alkmini Sgouritsa.
\newblock {A Little Charity Guarantees Almost Envy-Freeness}.
\newblock \emph{{SIAM} Journal on Computing}, 50\penalty0 (4):\penalty0 1336--1358, 2021.

\bibitem[Chaudhury et~al.(2024)Chaudhury, Garg, and Mehlhorn]{chaudhury2024efx}
Bhaskar~Ray Chaudhury, Jugal Garg, and Kurt Mehlhorn.
\newblock {EFX Exists for Three Agents}.
\newblock \emph{Journal of the ACM}, 71\penalty0 (1):\penalty0 1--27, 2024.

\bibitem[Christodoulou et~al.(2023)Christodoulou, Fiat, Koutsoupias, and Sgouritsa]{christodoulou2023fair}
George Christodoulou, Amos Fiat, Elias Koutsoupias, and Alkmini Sgouritsa.
\newblock {Fair allocation in graphs}.
\newblock In \emph{Proceedings of the 24th {ACM} Conference on Economics and Computation ({EC})}, pages 473--488, 2023.

\bibitem[Ebadian and Shah(2025)]{ebadian2025bit}
Soroush Ebadian and Nisarg Shah.
\newblock Every bit helps: Achieving the optimal distortion with a few queries.
\newblock In \emph{Proceedings of the 39th Annual AAAI Conference on Artificial Intelligence ({AAAI})}, 2025.

\bibitem[Farhadi et~al.(2021)Farhadi, Hajiaghayi, Latifian, Seddighin, and Yami]{farhadi2021almost}
Alireza Farhadi, MohammadTaghi Hajiaghayi, Mohamad Latifian, Masoud Seddighin, and Hadi Yami.
\newblock Almost envy-freeness, envy-rank, and {Nash} social welfare matchings.
\newblock In \emph{Proceedings of the 35th AAAI Conference on Artificial Intelligence ({AAAI})}, pages 5355--5362, 2021.

\bibitem[Feige(2025)]{Feige2025low}
Uriel Feige.
\newblock {Low communication protocols for fair allocation of indivisible goods}.
\newblock In \emph{Proceedings of the 26th {ACM} Conference on Economics and Computation ({EC})}, pages 358--382, 2025.

\bibitem[Garg and Murhekar(2023)]{garg2023computing}
Jugal Garg and Aniket Murhekar.
\newblock Computing fair and efficient allocations with few utility values.
\newblock \emph{Theoretical Computer Science}, 962:\penalty0 113932, 2023.

\bibitem[Garg et~al.(2025)Garg, Murhekar, and Qin]{garg2025constant}
Jugal Garg, Aniket Murhekar, and John Qin.
\newblock {Constant-Factor EFX Exists for Chores}.
\newblock In \emph{Proceedings of the 57th Annual ACM Symposium on Theory of Computing ({STOC})}, pages 1580--1589, 2025.

\bibitem[Halpern and Shah(2021)]{halpern2021fair}
Daniel Halpern and Nisarg Shah.
\newblock {Fair and Efficient Resource Allocation with Partial Information}.
\newblock In \emph{Proceedings of the 30th International Joint Conference on Artificial Intelligence ({IJCAI})}, pages 224--230, 2021.

\bibitem[Hv et~al.(2025)Hv, Ghosal, Nimbhorkar, and Varma]{hv2025efx}
Vishwa~Prakash Hv, Pratik Ghosal, Prajakta Nimbhorkar, and Nithin Varma.
\newblock Efx exists for three types of agents.
\newblock In \emph{Proceedings of the 26th ACM Conference on Economics and Computation ({EC})}, pages 101--128, 2025.

\bibitem[HV et~al.(2025)HV, Mehta, and Nimbhorkar]{hv2025almost}
Vishwa~Prakash HV, Ruta Mehta, and Prajakta Nimbhorkar.
\newblock Almost and approximate {EFX} for few types of agents.
\newblock \emph{CoRR}, abs/2508.15380, 2025.

\bibitem[Jin and Tao(2025)]{jin2025pareto}
Jiarong Jin and Biaoshuai Tao.
\newblock {On Pareto-Optimal and Fair Allocations with Personalized Bi-Valued Utilities}.
\newblock In \emph{Proceedings of the 21st Conference on Web and Internet Economics ({WINE})}, 2025.

\bibitem[Kaviani et~al.(2025)Kaviani, Keshavarz, Seddighin, and Shahrezaei]{kaviani2025improved}
Alireza Kaviani, Alireza Keshavarz, Masoud Seddighin, and AmirMohammad Shahrezaei.
\newblock {Improved Approximate {EFX} Guarantees for Multigraphs}.
\newblock \emph{CoRR}, abs/2506.09288, 2025.

\bibitem[Li et~al.(2022)Li, Bei, and Yan]{li2022proportional}
Zihao Li, Xiaohui Bei, and Zhenzhen Yan.
\newblock {Proportional allocation of indivisible resources under ordinal and uncertain preferences}.
\newblock In \emph{Proccedings of the 38th Conference on Uncertainty in Artificial Intelligence ({UAI})}, pages 1148--1157, 2022.

\bibitem[Lin et~al.(2025)Lin, Wu, and Zhou]{lin2025approximately}
Zehan Lin, Xiaowei Wu, and Shengwei Zhou.
\newblock {Approximately EFX and fPO Allocations for Bivalued Chores}.
\newblock In \emph{Proceedings of the 34th International Joint Conference on Artificial Intelligence ({IJCAI})}, pages 3952--3960, 2025.

\bibitem[Lipton et~al.(2004)Lipton, Markakis, Mossel, and Saberi]{lipton2004approximately}
Richard~J. Lipton, Evangelos Markakis, Elchanan Mossel, and Amin Saberi.
\newblock On approximately fair allocations of indivisible goods.
\newblock In \emph{Proceedings of the 5th ACM Conference on Electronic Commerce ({EC})}, pages 125--131, 2004.

\bibitem[Markakis and Santorinaios(2023)]{markakis2023improv}
Evangelos Markakis and Christodoulos Santorinaios.
\newblock {Improved {EFX} Approximation Guarantees under Ordinal-based Assumptions}.
\newblock In \emph{Proceedings of the 2023 International Conference on Autonomous Agents and Multiagent Systems ({AAMAS})}, page 591–599, 2023.

\bibitem[Oh et~al.(2021)Oh, Procaccia, and Suksompong]{oh2021fairly}
Hoon Oh, Ariel~D Procaccia, and Warut Suksompong.
\newblock {Fairly Allocating Many Goods with Few Queries}.
\newblock \emph{SIAM Journal on Discrete Mathematics}, 35\penalty0 (2):\penalty0 788--813, 2021.

\bibitem[Procaccia and Rosenschein(2006)]{procaccia2006distortion}
Ariel~D. Procaccia and Jeffrey~S. Rosenschein.
\newblock {The Distortion of Cardinal Preferences in Voting}.
\newblock In \emph{Proceedings of the 10th International Workshop on Cooperative Information Agents ({CIA})}, pages 317--331, 2006.

\bibitem[Sgouritsa and Sotiriou(2025)]{sgouritsa2025existence}
Alkmini Sgouritsa and Minas~Marios Sotiriou.
\newblock {On the Existence of EFX Allocations in Multigraphs}.
\newblock In \emph{Proceedings of the 24th International Conference on Autonomous Agents and Multiagent Systems ({AAMAS})}, pages 2735--2737, 2025.

\bibitem[Steinhaus(1949)]{steinhaus1949division}
Hugo Steinhaus.
\newblock Sur la division pragmatique.
\newblock \emph{Econometrica}, pages 315--319, 1949.

\bibitem[Sun et~al.(2023)Sun, Chen, and Doan]{sun2023fairness}
Ankang Sun, Bo~Chen, and Xuan~Vinh Doan.
\newblock Fairness criteria for allocating indivisible chores: connections and efficiencies.
\newblock \emph{Autonomous Agents and Multi Agent Systems}, 37\penalty0 (2):\penalty0 39, 2023.

\bibitem[Zeng and Mehta(2025)]{zeng2025structure}
Jinghan~A. Zeng and Ruta Mehta.
\newblock {On the Structure of EFX Orientations on Graphs}.
\newblock In \emph{Proceedings of the 24th International Conference on Autonomous Agents and Multiagent Systems ({AAMAS})}, pages 2309--2316, 2025.

\end{thebibliography}

\appendix

\section{Full Version of the \MFRR Algorithm}

\begin{algorithm}[H]
\caption{\MFRR (full version)}\label{alg:MF&RR:full}
\begin{algorithmic}[1]
\small
\State $M \gets \varnothing$  \Comment{Agents with known valuation function that contains a transition point.}
\State $R \gets \varnothing$  \Comment{Agents with possibly unknown valuation function with same value for all top-$n$ goods.}
\For{each agent $i \in N$} 
    \State Binary search over top-$n$ of $\succ_i$ for two consecutive goods $g$ and $g'$ such that $v_i(g) > v_i(g')$.
    \If{the search is successful}
        \State $M \gets M \cup \{i\}$ \Comment{$i$ becomes part of $M$.}
        \State $F_i \gets 0$  \Comment{and has an associated {\em freeze counter}, initially $0$.}
    \Else
        \State $R \gets R \cup \{i\}$ \Comment{$i$ becomes part of $R$.}
    \EndIf
    \State $X_i \gets \varnothing$ \Comment{Empty initial allocation.}
\EndFor
\State $\calP \gets G$ \Comment{The pool of available goods, initially all are available.}
\While{$\calP \neq \varnothing$}
\State $U \gets \varnothing$ \Comment{The set of unfrozen agents in $M$ for this round that will be assigned a good.}
\For{each agent $i \in M$}
    \If{$F_i > 0$}
        \State $F_i \gets F_i -1$ \Comment{If agent $i$ is frozen in this round, reduce the freeze counter by one.}
    \Else
        \State $U \gets U \cup \{i\}$ \Comment{If agent $i$ is not frozen, add her to set $U$.}
    \EndIf
\EndFor
\State $\mu \gets $ (prioritized) maximum matching between agents in $U$ and their high-valued goods.
\For{each agent $i \in \mu$}
    \State $X_i \gets X_i \cup \{\mu(i)\}$  \Comment{$i$ is assigned the good $\mu(i)$ she is matched to in $\mu$}
    \State $\calP \gets \calP \setminus \{\mu(i)\}$ \Comment{and this good is no longer available.}
\EndFor
\For{each agent $i \not\in \mu$}
\If{$\calP \neq \varnothing$}
   \State $X_i \gets X_i \cup \{$ arbitrary good $g \in \calP\}$ \Comment{$i$ is assigned an arbitrary available good}
   \State $\calP \gets \calP \setminus \{g\}$ \Comment{and this good is no longer available.}
   \For {each agent $j \in A\cap\mu$}
        \If {$v_i(\mu(j))=h_i$}             \Comment{If $j$ ``stole'' a high-value good from $i$, $j$ must freeze}
            \State $F_j \gets \lfloor h_i/\ell_i -1 \rfloor$  \Comment{for this many rounds.}
        \EndIf
   \EndFor
\EndIf
\EndFor
\For{each agent $i\in R$}
    \If{$\calP \neq \varnothing$}
        \State $X_i \gets X_i \cup \{ \text{top-ranked good } g \in \calP \text{ of } i \}$   \Comment{assign $i$ her top-ranked available good}
        \State $\calP \gets \calP \setminus \{g\}$             \Comment{and make it unavailable.}
    \EndIf
\EndFor

\EndWhile
\State \Return $\bX = (X_1,\ldots,X_n)$
\end{algorithmic}
\end{algorithm}

\end{document}